\documentclass[1p,times,authoryear,preprint]{elsarticle}
\usepackage{mathtools}

\usepackage{latexsym}              
\usepackage{amssymb}               
\usepackage{amsmath}               
\usepackage{amsthm}                
\usepackage{color}
\usepackage{enumitem}
\usepackage{comment}
\usepackage{marginnote}
\usepackage{hyperref}

\newcommand{\dep}{{\sf Dep}}

\newcommand{\QuoStar}{{\rm Quo}^\ast}
\newcommand{\RemStar}{{\rm Rem}^\ast}
\newcommand{\CedStar}{C^\ast}

\makeatletter
\def\iddots{\mathinner{\mkern1mu\raise\p@
\vbox{\kern7\p@\hbox{.}}\mkern2mu
\raise4\p@\hbox{.}\mkern2mu\raise7\p@\hbox{.}\mkern1mu}}
\makeatother

\newcommand{\col}{{\rm Col}}

\newcommand{\Quo}{{\rm Quo}}
\newcommand{\Frac}{{\rm Frac}}
\newcommand{\Rem}{{\rm Rem}}

\newcommand{\mylabel}[1]{\label{#1}}
\newcommand{\ind}[1]{\hspace*{#1em}}

\newcommand{\Note}{{\bf Note:\ }}

\newcommand{\Input}{{\bf Input:\ }}
\newcommand{\Output}{{\bf Output:\ }}

\newcommand{\If}{{\bf if\ }}
\newcommand{\Then}{{\bf then\ }}
\newcommand{\Else}{{\bf else\ }}

\newcommand{\Return}{{\bf return\ }}

\newcommand{\diag}{{\rm diag}}

\newcommand{\sign}{{\rm sign}}

\newcommand{\length}{{\rm length}}

\newcommand{\Z}{\ensuremath{\mathbb Z}}
\newcommand{\Q}{\ensuremath{\mathbb Q\mskip1mu}}

\newcommand{\B}{{\mathsf{B}}}
\newcommand{\M}{{\mathsf{M}}}

\newcommand{\Znn}{\ensuremath{\mathbb Z}^{n\times n}}

\newtheorem{theorem}{Theorem}

\newtheorem{definition}[theorem]{Definition}
\newtheorem{corollary}[theorem]{Corollary}

\newtheorem{remark}[theorem]{Remark}
\newtheorem{lemma}[theorem]{Lemma}
\newtheorem{example}[theorem]{Example}
\newtheorem{proposition}[theorem]{Proposition}

\DeclareMathOperator{\rowmod}{{\mathbf r}mod}
\DeclareMathOperator{\colmod}{{\mathbf c}mod}
\DeclareMathOperator{\loglog}{loglog}

\newlength{\algwidth}
\usepackage{natbib}

\usepackage{float}

\begin{document}
\begin{frontmatter}

\title{A fast algorithm for computing the Smith normal form
with multipliers for a nonsingular integer matrix}

\author{Stavros Birmpilis}
\address{Cheriton School of Computer Science, University of Waterloo,
Waterloo ON, Canada N2L 3G1}
\ead{sbirmpil@uwaterloo.ca}

\author{George Labahn}
\address{Cheriton School of Computer Science, University of Waterloo,
Waterloo ON, Canada N2L 3G1}
\ead{glabahn@uwaterloo.ca}

\author{Arne Storjohann}
\address{Cheriton School of Computer Science, University of Waterloo,
Waterloo ON, Canada N2L 3G1}
\ead{astorjoh@uwaterloo.ca}

\begin{abstract}
A Las Vegas randomized algorithm is given to compute the Smith
multipliers for a nonsingular integer matrix $A$, that is, unimodular
matrices $U$ and $V$ such that $AV=US$, with $S$  the Smith normal
form of $A$.  The expected running time of the algorithm is about
the same as required to multiply together two matrices of the same
dimension and size of entries as $A$.   Explicit bounds are given
for the size of the entries in both unimodular multipliers.  The
main tool used by the algorithm is the Smith massager, a relaxed
version of $V$, the unimodular matrix specifying the column operations
of the Smith computation.   From the perspective of efficiency, the
main tools used are fast linear system solving and partial linearization
of integer matrices.  As an application of the Smith with multipliers
algorithm, a fast algorithm is given to find the fractional part
of the inverse of the input matrix.
\end{abstract}


\end{frontmatter}


\section{Introduction}

Let $A\in\Znn$ be a nonsingular integer matrix with
$$
\setlength{\arraycolsep}{.5\arraycolsep} \renewcommand{\arraystretch}{.75}
S:=\diag(s_1,s_2,\ldots, s_n) = \left [ \begin{array}{cccc} s_1 & & & \\
 & s_2 & & \\
& & \ddots & \\
 & & & s_n \end{array} \right ] \in \Z^{n \times n}
$$ its Smith normal form. There are  unimodular matrices $U, V\in\Znn$
which describe the set of invertible integer row and column operations
which transform $A$ into its Smith form $S$ or vice versa. These
row and column operations are typically defined as satisfying matrix
equations in the form $ U A V = S$ or $A = U S V$. In our case, it
will be convenient to specify these Smith form multipliers as unimodular
matrices satisfying
\begin{equation} \mylabel{eq:AVUS1}
AV = US.
\end{equation}

\paragraph*{\bf Motivation}
In some cases, just knowing the Smith form is all that is
needed in applications. For example, to determine whether two integer matrices are
equivalent up to unimodular row and column operations,
it is enough to see if they have the same
Smith form.  Similarly, if $A$ is the relation matrix for a finite
abelian group $G$, then knowing its Smith form is enough to classify
the group into a direct sum of cyclic groups~\citep{Cohen,Newman}.
Such a classification in turn is used, for example, to efficiently
compute Gr\"obner bases of ideals invariant under the action of an
abelian group~\citep{FaugereSvartz}.

However there are applications where both the Smith form
and its unimodular multipliers are needed.
Consider for example the linear system solving problem
\begin{equation} \mylabel{eq:linsys}
x A = b,
\end{equation}
that is, given a row vector $b \in \Z^{1 \times n}$, find the unique row vector
$x \in \Q^{1 \times n}$ such that $xA=b$.   Using the representation~(\ref{eq:AVUS1}), we can transform the linear
system in~(\ref{eq:linsys}) to
$$
\bar{x} S = \bar{b},
$$
with $\bar{x} = xU$ and $\bar{b}=bV$.  Since $S$ is in Smith form,
the new system allows for easier determination of possible properties
of the solution.  For example, the denominator of $x$, the smallest
integer $d \in \Z_{>0}$ such that $dx$ is integral,
will be the same as the denominator of $\bar{x} = \bar{b}S^{-1}$.

The above example gives one application where both the Smith form
and its unimodular multipliers are needed. Smith multipliers are
also needed in a number of other settings. For example, when one
not only wants the classification of a finite abelian group into
the direct sum of its cyclic components, but also the isomorphism
which takes the group to the direct sum of cyclic factors.  If $x$
is a row vector whose entries are generators of an abelian group and
matrix $A$ represents the relations among the entries of $x$ such
that $xA=0$, then $\bar{x}=xU$ is a new set of generators with
relations simply given by $\bar{x}S=0$.  Both the Smith form and
its multipliers are needed when one looks for possible rational
symmetry by a finite abelian group action for a set of polynomial
equations along with determining the rational invariants and rewrite
rules of such an action~\citep{HubertLabahn}. Other applications
which make use of the Smith multipliers include determining lattice
rules for quadrature formulas over the unit cube~\citep{LynessKeast},
its use in chip-firing for finite connected graphs in
combinatorics~\citep{Stanley}, and many more.

\paragraph*{\bf Computation}
Initial algorithms for Smith form computation  \citet{smith,bradley70,Bradley}  were modelled on Gaussian elimination where greatest common divisors and the associated solutions of
linear diophantine equations replaced division. These early algorithms encountered rapid growth of intermediate computations. However, efficient computation of the Smith form  could make use of the fact that the diagonal elements are the invariant factors of the matrix, factors which can be represented as ratios of greatest common divisors of minors of the matrix. 
As the Smith form is unique one can for example use homomorphic imaging techniques~\citep{GeddesCzaporLabahn} for these computations.  The first algorithm to compute the Smith form with multipliers in polynomial time originated with \citet{KannanBachem}. The multipliers are not unique with \citet{thesis} being the first to consider the problem of small unimodular multipliers for Smith computation.

Let $\omega$ be a valid exponent of matrix multiplication: two $n\times n$ matrices can be multiplied in $O(n^{\omega})$ operations
from the domain of entries.  Furthermore, let $\|A\|$ denote the largest entry of $A$ in absolute value.  Recent
fast methods include that of \citet{KaltofenVillard04a} which
combines a Las Vegas algorithm for computing the characteristic
polynomial with ideas of \citet{Giesbrecht01}, to obtain a Monte
Carlo algorithm for the Smith form in time $(n^{3.2}\log \|A\|)^{1
+ o(1)}$ assuming $\omega = 3$, and in time $(n^{2.695591}\log
\|A\|)^{1+o(1)}$ assuming the currently best known upper bound
$\omega < 2.37286$ for $\omega$ by \citet{AlmanWilliams2021} and the best
known bound for rectangular matrix multiplication by \citet{LegallUrrutia18}.  


\paragraph*{\bf Our main contribution}
An important long-term program in exact linear algebra with polynomial or integer matrices is to obtain algorithms whose cost is about the same as for multiplying two matrices of corresponding dimension and entry sizes.  In the case of Smith form this was solved in~\citep{BirmpilisLabahnStorjohann20} which  gave
a Las Vegas algorithm for the Smith form in time $(n^\omega\log \|A\|)^{1+o(1)}$.   
However it was not yet known how one can obtain both the Smith form
and its multipliers in a similar complexity.
A major difficulty is that the bitlength of the entries in $U$ and
$V$ can be asymptotically larger than those in $A$. The previously fastest
algorithm
given in \citet{thesis}  recovers $U$ and $V$ in the form $UAV=S$
in time $(n^{\omega+1}\log \|A\|)^{1+o(1)}$.   

The main contribution in this paper is a new Las Vegas algorithm
which allows us to compute
$S$, $U$ and $V$ satisfying~(\ref{eq:AVUS1}) with
approximately the same number of bit operations
as required to multiply two
matrices of the same dimension and size of entries as the input
matrix. As we already have a fast way to compute the Smith form
$S$, our goal in this paper is an efficient algorithm that also
returns the unimodular matrices $U$ and $V$. Previously, determining
the Smith form alone had been considered easier than determining
the Smith form and it's multipliers. In this paper, we show that
finding the multipliers can be done in the same time as computing
the Smith form, at least in terms of asymptotic complexity. However,
finding the multipliers requires some new, novel ideas.

\paragraph*{\bf Our approach}
The  Las Vegas algorithm in \citet{BirmpilisLabahnStorjohann20} 
computes not only the Smith form $S$ but also
returns a {\em massager} matrix $M$. This matrix  satisfies the property that
$$
AM\equiv 0 \colmod S  \mbox{ and } W  M\equiv I_n \colmod S
$$
for some integer matrix $W$.  
Here, $\colmod$ denotes working modulo
columns: $B \equiv C \colmod S$ if column $j$ of $B$ is
congruent modulo $s_j$ to column $j$ of $C$, $1\leq j \leq n$.
On the one hand, 
a massager $M$ is in general not unimodular and thus
is a relaxed version of $V$ in the equations
$$
A V = U S  \mbox{ and }  V^{-1}  V = I_n,
$$
where $V^{-1}$ is integral since $V$ is unimodular.
On the other hand, a Smith multiplier $V$ is precisely a massager that \emph{is} unimodular.  Massagers
were introduced by \citet{BirmpilisLabahnStorjohann19} and
are the main tool used in this paper to efficiently compute the Smith
multipliers.
Our approach is to perturb a massager $M$ by a random
matrix $R$ scaled by the Smith form, that is, a matrix of the form
$\bar{M} := M + R S$.  We show that the perturbed matrix $\bar{M}$ remains
a massager. Moreover, we prove that with high probability
the perturbation has the effect that the submatrix comprised of the last $n-1$ columns of $\bar{M}$ will be primitive, that is, $\bar{M}$ will be  left equivalent to a nonsingular lower
triangular matrix $\bar{H}$ that has the shape
\begin{equation} \mylabel{eq:Htriv}
\bar{H} = \left [ \begin{array}{ccccc} |\det \bar{M}| & & & & \\
 \ast & 1 & & & \\
 \ast &  & 1 & & \\
 \vdots  & && \ddots & \\
\ast & & && 1 \end{array} \right ],\end{equation}
with all $\ast$ entries nonnegative and reduced modulo $|\det \bar{M}|$.
We remark that $\bar{H}$ is the unique lower triangular row Hermite
form of $\bar{M}$.  
In case the perturbation is successful and
$\bar{H}$ is trivial, that is, has the shape shown in~(\ref{eq:Htriv})
with all off-diagonal
entries except for possibly the first equal to one, then 
we give an algorithm to compute it quickly (or determine that it is not trivial and report {\sc Fail}).
Since $\bar{H}$ is left equivalent to $\bar{M}$,  the matrix
$V := \bar{M}\bar{H}^{-1}$ will not only be integral but 
also unimodular. Based on the structure $\bar{H}$ we can show that
$V$ is also a massager.  The matrix $V$ is then one of our Smith multipliers.
Exploiting again the fact that $\bar{H}$ is trivial, we show how to compute the 
product  $\bar{M} \bar{H}^{-1}$ efficiently.  The other multiplier $U$ is
constructed using  (\ref{eq:AVUS1}).

Our approach allows us to
establish explicit bounds on the size of the two
unimodular multipliers.  For example, if we define the bitlength of an integer
column vector to be bitlength of the maximum magnitude entry, then we can show that the average
bitlength of the columns  of either unimodular multiplier
matrix is bounded by $O(n(\log n + \log \|A\|))$. The overall size (the sum of the bitlengths of
all of the entries) of either multiplier matrix is then bounded by $O(n^2(\log n + \log \|A\|))$.

\paragraph*{\bf Additional contributions}
In order to obtain the desired running time for our algorithm we need to extend a
some previously known algorithms to a slightly more general setting.

Our first additional contribution is to give extensions of subroutines
for \emph{linear system solving} and \emph{integrality certification}.
We briefly recall what these two problems are.  Given an
integer matrix $B$ with the same number of rows as $A$, together
with an integer lifting modulus $X \in \Z_{>0}$ that is relatively
prime to $\det A$, the linear system solving problem is to compute
$\Rem(A^{-1}B,X^d)$ for a given precision $d$. 
Here, $\Rem(a,X)$ for an integer $a$ and
positive integer $X$ denotes the unique integer in the range $[0,X-1]$
that is congruent to $a$ modulo $X$.  If the first argument of
$\Rem$ is a matrix or vector, the function applies element-wise.
The integrality
certification problem is to determine if $A^{-1}B$ is integral.
\citet{BirmpilisLabahnStorjohann19} use the double-plus-one
lifting approach of \citet{PauderisStorjohann12} to obtain 
a fast algorithm for the linear system solving problem.
\citet{BirmpilisLabahnStorjohann20} follows this up with
a fast algorithm for integrality certification. Both of
the algorithms mentioned above were analyzed only in the special case when $X$
is a power of 2, thus requiring the hypothesis that $\det A$ is an odd integer.
In Section~\ref{sec:prelim} we extend the linear system solving and integrality certification algorithms in
\citep{BirmpilisLabahnStorjohann19, BirmpilisLabahnStorjohann20}
to the case where $X$ is the power of a small prime, thus allowing
to handle the case of input matrices $A$ with arbitrary determinant.

Our second additional contribution is to extend partial
linearization techniques previously developed for polynomial matrices
to the integer setting.  The cost of algorithms on an integer
matrix $A$ are typically sensitive to $\log ||A||$, the maximum
bitlength of the entries. If only some entries have large bitlength,
for example the  average bitlength of the rows or columns is
small, then for many problems partial linearization can be used to
transform to a new problem on an input matrix that has maximum bitlength
of entries the average bitlength of the rows or columns of the
original.  Section~\ref{sec:lin} extends the partial linearization
technique of \citet[Section~6]{GuptaSarkarStorjohannValeriote11}
for polynomial matrices to the integer setting,  and gives applications
to a number of problems. In particular this includes the linear system solving
and integrality certification problems discussed above.

Our final contribution is to resolve an open question from
\cite{Storjohann10a}, which asks if one can compute the proper
fractional part of $A^{-1}$ while avoiding any dependence on 
$\log ||A^{-1}||$. Note that $\log \|A^{-1}\|$ 
is a measure of how much larger the bitlength of numerators in $A^{-1} \in \Q^{n \times n}$
are compared to their respective denominators. (If $\log \|A^{-1}\| < 0$ then all
entries in $A^{-1}$ are proper fractions, but it is possible that $\log\|A^{-1}\| \in \Omega(n(\log n + \log \|A\|))$,
for example if $A$ is unimodular.)
Recall the notion of the proper fractional
part of $A^{-1}$.  Let $s \in \Z_{>0}$ be the largest entry
in the Smith form of $A$. Then $s$ is the minimal integer such that $sA^{-1}$
is integral.  The proper fractional part of $A^{-1}$ is then
$\Rem(sA^{-1},s)/s$.  
To computing the proper fractional part of $A$ it is thus sufficient
to compute $\Rem(sA^{-1},s)$.

\cite{Storjohann10a} computes $\Rem(sA^{-1},s)$ by first computing
an \emph{outer product adjoint formula}
for $A$: a triple of matrices $(\bar{V}, S, \bar{U})$ such that
\[\Rem(sA^{-1}, s) = \Rem(\bar{V}(sS^{-1})\bar{U}, s).  \]
There is a direct relationship between an outer product adjoint
formula and the unimodular Smith multipliers $U$ and $V$. Using
this relationship, and as an application of our work, we show
in Section~\ref{sec:opa} that an outer product formula
can be computed in time $(n^{\omega} \log \|A\|)^{1+o(1)}$
bit operations.  This improves on the algorithm 
of \citep{Storjohann10a} by incorporating
fast matrix multiplication and removing any
dependence of the complexity
on $\log \|A^{-1}\|$ in case $\|A^{-1}\| > 1$.

\paragraph*{\bf Organization of the paper}
The remainder of this paper is organized as follows. Section~\ref{sec:sm}
defines our main tool, the Smith massager of a nonsingular integer
matrix, and gives several important properties. Section~\ref{sec:prelim}
gathers together a collection of computational tools related to linear system
solving which we will require for our main algorithm. Section~\ref{sec:lin}
presents a partial linearization technique which, in many algorithms,
helps us replace the dependency of the cost estimates on the bit
length of the largest entry of the input with the average bit length.
Section~\ref{sec:informal} gives a high-level description of our
main algorithm for computing Smith multipliers using an example.
Section~\ref{sec:random} proves the main probabilistic argument of
our process, namely, the fact that a randomly perturbed Smith
massager has an almost trivial Hermite form. Sections~\ref{sec:herm}
and~\ref{sec:algo} present the main algorithm and rigorously
prove the claimed time complexity along with bounds on the sizes
of the multipliers. Section~\ref{sec:opa} shows how we can apply
the Smith multiplier matrices in order to obtain an outer adjoint
formula along with its complexity. The paper ends with a conclusion
and topics for future research.

\paragraph*{\bf Cost model}

Following~\cite[Section~8.3]{vonzurGathenGerhard}, cost estimates
are given using a function $\M(d)$ that bounds the number of bit
operations required to multiply two integers bounded in magnitude
by $2^d$.  We use $\B(d)$ to bound the cost of integer gcd-related
computations such as the extended euclidean algorithm.  We can
always take $\B(d) = O(\M(d) \log d)$. If $\M(d) \in
\Omega(d^{1+\epsilon})$ for some $\epsilon >0$ then $\B(d) \in
O(\M(d))$.

As usual, we assume that $\M$ is superlinear and subquadratic.  We
also assume that $\M(ab) \in O(\M(a)\,\M(b))$ for $a,b \geq 1$.  We
assume that $\omega>2$, and to simplify cost estimates we make the
assumption that $\M(d) \in O(d^{\omega-1})$. This assumption simply
stipulates that if fast matrix multiplication techniques are used,
then fast integer multiplication techniques should also be used.
The assumptions stated in this paragraph apply also to $\B$.

\section{Smith massagers} \label{sec:sm}

In this section we introduce our main tool, the {\em Smith massager}
of a nonsingular integer matrix $A \in \Z^{n \times n}$.  We provide
the definition and basic features and identify some matrix operations
that keep the massager properties intact.  In
Subsection~\ref{ssec:compinv}, we show how the Smith massager gives
an alternative, compact representation of the lattice $\{ v A \mid
v \in \Z^{n\times n}\}$, the set of all $\Z$-linear combinations
of the rows of $A$.  Finally, in Subsection~\ref{ssec:appruni}, we
present additional properties of massagers which will help us to
compute Smith multipliers.

\begin{definition} \label{def:SM}
Let $A\in\Znn$ be a nonsingular integer matrix with Smith form $S$.
A matrix $M\in\Znn$ is a \emph{Smith massager} for $A$ if
\begin{itemize}
\item[(i)] it satisfies that
\begin{equation} \label{eq:AM}
AM\, \equiv\, 0\, \colmod S, and
\end{equation}
\item[(ii)] there exists a matrix $W\in\Znn$ such that
\begin{equation}\label{eq:UM}
W  M\, \equiv\, I_n\, \colmod S.
\end{equation}
\end{itemize}
\end{definition}

Property (i) of a Smith massager $M$ implies that the matrix
$AMS^{-1}$ is integral, while property (ii) implies that $M$ is
unimodular up to modulo the columns of $S$. Thus, matrix $M$ acts like the
multiplier matrix $V$ in $AV=US$ except that it
relaxes the unimodularity property. Our objective will be to transform
$M$ to a new Smith massager that is in fact unimodular over the integers. Note that
any Smith massager reduced column modulo $S$ is still a Smith
massager.  If $M= (M\colmod S)$, then  $M$ is called a \emph{reduced
Smith massager}.
We remark that a reduced massager can be be represented with only $O(n^2(\log n+\log \|A\|))$ bits.

\begin{example} \label{ex:smithmas}
The Smith form of
\[A = \left[ \begin{array}{cccc}
-6& 3& -13& -15\\
-4& 19& 12& -1\\
-4& 10& -6& 17\\
-26& -13& 1& -2
\end{array} \right]\]
is $S = \diag(1, 1, 9, 29088)$. For
\[ M = \left[ \begin {array}{cccc} 0&0&7&805\\ 0&0&5&23668
\\ 0&0&3&6\\ 0&0&4&10224
\end {array} \right],\]
we have $AM \equiv 0 \colmod S,$ while setting
\[W = \left[ \begin {array}{cccc} 4&-19&-12&1\\ -306&3&
133&0\\ 5156&805&6332&0\\ 12017&-
403&11356&0\end {array} \right]\]
gives
\[W M = I_4 + \left[ \begin {array}{cccc} -1&0&-99&-436320\\ 0&-1
&-1728&-174528\\ 0&0&59112&23241312
\\ 0&0&116172&203616\end {array} \right]
\left[ \begin{array}{cccc}
1&&&\\ &1&&\\
&&9&\\ &&&29088
\end{array} \right],\]
implying that $WM \equiv I_4 \colmod S$. It follows that $M$ is a
Smith massager for $A$.
\end{example}

It will be useful to notice that a Smith massager $M$ for some matrix $A$ remains a valid Smith massager under some specific columns operations.

\begin{lemma} \mylabel{lem:smops}
Assume $M\in\Znn$ is a Smith massager for $A$. Then the  matrix obtained from $M$ by
\begin{itemize}
\item[(i)]  adding any integer column vector multiplied by $s_i$ to column $i$,
\item[(ii)]  adding any multiple of a latter to a former column, or
\item[(iii)] multiplying (or dividing exactly) the $i^{th}$ column by an integer relatively prime to $s_i$
\end{itemize}
is also a Smith massager for $A$.
\end{lemma}
\begin{proof}
For each one of these operations, we need to show that the modified
matrix $M$ still satisfies properties (i) and (ii) of
Definition~\ref{def:SM}.

Let $\bar{M}$ be the matrix obtained from $M$ by performing operation
(i).  Then $\bar{M} \equiv M \colmod S$ and thus $A\bar{M} \equiv
0 \colmod S$ and $W\bar{M} \equiv I_n \colmod S$ still hold.

For operation (ii), let $1 \leq  i_1 < i_2 \leq n$ and $c \in \Z$.
Let $\bar{M}$ be the matrix obtained from $M$ by adding $c$ times
column $i_2$ to column $i_1$. Because $s_{i_1} \mid s_{i_2}$,
$A\bar{M} \equiv 0 \colmod S$ still holds.
Let $\bar{W}$ be the matrix obtained from $W$ by adding
$-c$ times row $i_1$ to row $i_2$. Then $\bar{W}\bar{M} \equiv I_n \colmod S$.
 
For  operations (iii), let $c \in \Z$ be relatively prime to $s_i$.
Let $\bar{M}$ be the matrix obtained from $M$ by multiplying column $i$
by $c$. Then $A\bar{M} \equiv 0 \colmod S$ still holds.
Let $\bar{W}$ be the matrix obtained from $M$ by multiplying
row $i$ by $\Rem(1/c,s_n) \in \Z$.  Then $\bar{W}\bar{M} \equiv I_n \colmod S$.
The case for $1/c$ is similar.
\end{proof}

\subsection{Alternate characterizations of the lattice $\{ vA \mid v \in \Z^{1 \times n} \}$} \mylabel{ssec:compinv}

Let $A \in \Z^{n \times n}$ be nonsingular.  The set of all $\Z$-linear
combinations of the rows of $A$ generates the integer lattice $\{
vA \mid v \in \Z^{1 \times n}\}$.  The following theorem gives
alternate characterizations of the same lattice which will be useful
in Section~\ref{sec:herm} to give an compact description of the Hermite
form of $A$ in terms of a Smith  massager for $A$.

\begin{theorem} \mylabel{thm:msproj}
Let $A\in\Znn$ be nonsingular with Smith form $S$ and Smith massager $M$.
Let $s$ be the largest invariant factor of $S$. The following lattices
are all identical:
\begin{itemize}
\item  $L_1 =  \{ v A \mid v \in \Z^{1 x n} \}$
\item  $L_2 =  \{  v  \mid v A^{-1} \in \Z^{1 \times n} \}$
\item  $L_3 =  \{  v  \mid v MS^{-1} \in \Z^{1 \times n} \}$
\item  $L_4 =  \{  v  \mid v M(sS^{-1}) \equiv 0_{1 \times n} \bmod s \}$
\item  $L_5 =  \{  v  \mid v M \equiv 0_{1 \times n} \colmod S \}$
\end{itemize}
\end{theorem}

\begin{proof}
It is straightforward to show that $L_1=L_2$, $L_3=L_4$ and $L_4=L_5$ by 
verifying that each of these pairs of sets are subsets of each other. To complete the
proof it will be sufficient to show that 
$L_2=L_3$.

Let
\[
B=\left [ \begin{array}{cc} A & \\
 & I_n \end{array} \right ]
\left [ \begin{array}{cc} I_n & \\
-W & I_n \end{array} \right ]
\left [ \begin{array}{cc} I_n & M \\
  & I_n \end{array} \right ]
\left [ \begin{array}{cc} & I_n \\
 S^{-1} & \end{array} \right ]
= 
\left [ \begin{array}{cc} AM S^{-1} & A \\
( I_n-WM) S^{-1} & -W \end{array} \right ].
\]
By Definition~\ref{def:SM}  $B$ is integral.
Furthermore, since $|\det A|=\det S\neq 0$, 
$B$ is unimodular. If we premultiply $B$ by $\diag(A^{-1},
I_n)$ and then restrict to the first $n$ rows, we obtain
\begin{equation} 
\left [ \begin{array}{cc} A^{-1} & \end{array} \right ]B=\left [ \begin{array}{cc} MS^{-1} & I_n \end{array} \right ].
\end{equation}
Since both $B$ and $B^{-1}$ are integral, we conclude that for any
$v\in\Z^{1\times n}$, $vA^{-1}$ is integral if and only if $vMS^{-1}$
is integral.  It follows that $L_2=L_3$.
\end{proof}

The following corollary follows from the equality of $L_2$ and $L_3$ in Theorem~\ref{thm:msproj}.

\begin{corollary} \mylabel{cor:msdenom}
Let $A\in\Znn$ be nonsingular with Smith form $S$ and Smith massager $M$. For any row vector $v\in\Z^{1\times n}$, the denominator of $vA^{-1}$ equals the denominator of $vMS^{-1}$.
\end{corollary}

As remarked earlier, if $M$ is a reduced massager, then  $MS^{-1}$ can be represented
with only $O(n^2(\log n+\log \|A\|))$ bits. This compares to
$O(n^3(\log n+\log \|A\|))$ bits required for $A^{-1}$.

\begin{example}
Matrix
\[A = \left[ \begin{array}{cccc}
-6& 3& -13& -15\\
-4& 19& 12& -1\\
-4& 10& -6& 17\\
-26& -13& 1& -2
\end{array} \right],\]
from Example~\ref{ex:smithmas}, has Smith form $S = \diag(1, 1, 9, 29088)$ and Smith massager
\[ M = \left[ \begin {array}{cccc} 0&0&7&805\\ 0&0&5&23668
\\ 0&0&3&6\\ 0&0&4&10224
\end {array} \right].\]
In this case,
\[ A^{-1} = \frac{1}{29088}\left[ \begin{array}{cccc}
-271& -402& -373& -937\\
580& 920& 524& -356\\
-1074& 804& -870& 258\\
-784& -352& 1008& 80
\end{array} \right],\]
and from Corollary~\ref{cor:msdenom}, for any row vector $v\in\Z^{1\times n}$, the denominator of $vA^{-1}$ equals the denominator of
\[ v\left[ \begin{array}{cc}
7& 805\\
5& 23668\\
3& 6\\
4& 10224
\end{array} \right]
\left[ \begin{array}{cc}
1/9&\\ &1/29088
\end{array} \right],\]
where the first two columns can be omitted because the corresponding invariant factors are $1$. 
Equivalently, 
from the equality of $L_3$ and $L_5$ in Theorem~\ref{thm:msproj}, we have that
\[\left[ \begin{array}{cccc}
-271& -402& -373& -937\\
580& 920& 524& -356\\
-1074& 804& -870& 258\\
-784& -352& 1008& 80
\end{array} \right] \equiv_R \left[ \begin{array}{cc}
7& 805\\
5& 23668\\
3& 6\\
4& 10224
\end{array} \right]
\left[ \begin{array}{cc}
3232&\\ &1
\end{array} \right] \bmod 29088.\]
\end{example}

Recall that a basis for the lattice $L_1$ in Theorem~\ref{thm:msproj} is any matrix
that is left equivalent to $A$, for example $A$ itself.
The following theorem follows from the equality of $L_1$ and $L_5$
in Theorem~\ref{thm:msproj}.


\begin{theorem} \mylabel{thm:canms}
Let $A\in\Znn$ be nonsingular with Smith form $S$ and a Smith massager $M$. A matrix 
$H\in\Znn$ is left equivalent to $A$ if and only if $|\det H| = \det S$ and $HM\equiv 0\colmod S$.
\end{theorem}

In other words, the Smith form $S$ and a Smith massager $M$ can be
used to describe a left equivalent canonical form of a matrix $A$
in a compact and fraction-free way. We will use Theorem~\ref{thm:canms}
later in Section~\ref{sec:herm}.

\subsection{Creating a unimodular Smith massager}
\mylabel{ssec:appruni}

Let $A \in \Z^{n \times n}$ be nonsingular.  In this subsection we
give a high level overview of our algorithm to produce a Smith
multiplier $V$ such that $AV=US$.  Recall that a Smith multiplier $V$ is precisely a
Smith massager that is unimodular. Once $V$ has been found we recover $U$ as $U := AVS^{-1}$.
Our approach to computing a unimodular $V$ has four steps:
\begin{enumerate}
\item Compute the Smith form $S$ and a reduced Smith massager
$M$ for $2A$.
\item Choose a random perturbation matrix $R \in \Z^{n\times n}$ and
let $\bar{M} := M + 2RS$.
\item Compute the lower triangular row Hermite form $H$
of $\bar{M}$.
\item Return $V := \bar{M} H^{-1}$.
\end{enumerate}
The reason, in step~1, for computing a Smith massager $M$
for $2A$ instead of $A$ is that 
matrix $\bar{M}$ produced in step~2 will be a nonsingular, independent
of the choice of $R$.  The purpose of the perturbation in step~2
is to ensure, with high probability, that $\bar{M}$ has a trivial
lower triangular Hermite form, that is, with all but possibly the
first diagonal entry equal to~1.  Knowing {\it a priori} that $\bar{M}$ is 
nonsingular simplifies our deriviation of a lower bound on the probability 
the Hermite form $H$ of $\bar{M}$
has at most one non-trivial column.  Having $H$ be trivial is
important for the efficiency of steps~3 and~4, and also to obtain
good bounds on the size of entries of $V$.


Filling in the details of how to choose $R$ in step~2 and how to
do each of the steps efficiently is the main topic of the rest of
this article.  Section~\ref{sec:prelim} gathers together required subroutines
related to linear system solving, and in particular shows that
step~1 can be done efficiently.  Section~\ref{sec:lin} develops a
partial linearization technique which allows to efficiently compute
with matrices with entries of skewed bitlength, for example
the matrix $\bar{M}$ in step~2 which has columns of skewed bitlength.
Section~\ref{sec:informal}
then gives a worked example of the above four step algorithm and
points to Sections~\ref{sec:random}--\ref{sec:algo} for 
algorithms to perform steps 3--4 efficiently.

In the remainder of this subsection, our goal is only to establish that the above
recipe is correct, namely, that the matrix $V$ returned in step~4
will be a unimodular Smith massager, independent of the choice of
$R$ in step~2. To do this, we need to establish that: (a) $M$  in step~1
is a nonsingular Smith massager of $A$ even though it is computed to be a Smith massager for $2A$;
(b) $\bar{M}$ in step~2 remains a nonsingular Smith massager for $A$, despite the
additive perturbation $+2RS$, and independent of choice of $R$; (c) the
matrix $V$ produced in step~4 is a Smith massager for $A$.
On the on hand, the fact that $V$ produced in step~4 is unimodular is straightforward:
$H$ is left equivalent to $\bar{M}$ and so $\bar{M}H^{-1}$ will be integral with 
determinant $\pm 1$. On the other hand, what we need to prove in step~4 is that the column
operations effected by the postmultiplication of $H^{-1}$ in $V := \bar{M}H^{-1}$ 
always produces a $V$ that is a Smith massager of $A$.

\begin{proposition} \mylabel{prp:cARS}
Let  $c\in\Z_{>0}$ and $A\in\Znn$. If $M\in\Znn$ is a Smith massager for $cA$, then for any matrix $R\in\Znn$:
\begin{itemize}
\item[(i)] $M+R(cS)$ is a Smith massager for $A$.
\item[(ii)] The last $i$ columns of $M+R(cS)$ have full rank over $\Z/(p)$ for any prime
$p$ that divides $(cs_{n-i+1})$. 
\end{itemize}
\end{proposition}
An immediate corollary of Proposition~\ref{prp:cARS} is that a Smith
massager for $2A$ will be a nonsingular Smith massager of $A$.  The
proof of Proposition \ref{prp:cARS} follows directly from the next
two lemmas and Definition \ref{def:SM}.

\begin{lemma} \mylabel{lem:cA}
Let $c\in\Z_{>0}$ and $A\in\Znn$.  If $M\in\Znn$ is a Smith massager for $cA$, then $M$ is also a Smith massager for $A$.
\end{lemma}

\begin{proof}
First note that if $S\in\Znn$ is the Smith form of $A$, then $cS$ is the Smith form of $cA$. Since $M$ is a
Smith massager for $cA$, Definition~\ref{def:SM} states that
\begin{equation} \mylabel{eq:part1}
cAM \equiv 0 \colmod cS,
\end{equation}
and that there exists a $W \in \Z^{n\times
n}$ such that 
\begin{equation} \mylabel{eq:part2}
WM \equiv I_n \colmod cS.
\end{equation}
It follows from (\ref{eq:part1}) that $AM \equiv 0 \colmod S$ and
from (\ref{eq:part2}) that $WM \equiv I_n \colmod S$, and thus by
Definition~\ref{def:SM}, $M$ is a Smith massager for $A$.
\end{proof}

\begin{lemma} \mylabel{lem:pdivs}
For any prime $p$ that divides $s_{n-i+1}$, the last $i$ columns of a Smith massager $M$ have full rank over $\Z/(p)$.
\end{lemma}

\begin{proof}
The claim follows from Definition~\ref{def:SM} of the Smith massager since
\begin{align*}
WM &\equiv I_n \colmod \setlength{\arraycolsep}{.5\arraycolsep} \renewcommand{\arraystretch}{.75} \left[\begin{array}{ccc} s_1 && \\ &\ddots & \\ && s_n \end{array}\right] \\
&\equiv I_n \colmod \setlength{\arraycolsep}{.5\arraycolsep} \renewcommand{\arraystretch}{.75} \left[\begin{array}{cccccc} s_1 &&&&& \\ &\ddots &&&& \\ && s_{n-i} &&& \\ &&& p && \\ &&&& \ddots & \\ &&&&& p \end{array}\right].
\end{align*}
If the last $i$ columns of $WM\bmod p$ have full rank, then  the last $i$ columns of $M\bmod p$ also have full rank.
\end{proof}

Now consider steps 3 and 4 of the recipe. 
The \emph{lower triangular row Hermite form} of a nonsingular matrix
$A\in\Znn$ is the unique matrix
\[H := \left [\begin{array}{cccc}
h_1 &&& \\
\ast & h_2 && \\
\vdots & \vdots & \ddots & \\
\ast & \ast & \cdots & h_n
\end{array}\right ]\in\Znn\]
that is left equivalent to $A$, has positive diagonal entries,
and has off-diagonal entries in each column reduced by the
diagonal entry in the same column.  Lemma~\ref{lem:masher2} provides the final
ingredient to establish the correctness of our recipe
by proving that a nonsingular
Smith massager for $A$, post-multiplied by the inverse of its lower
triangular row Hermite form, is still a Smith massager for $A$.  Lemma~\ref{lem:masher} is an intermediate
result.

\begin{lemma} \mylabel{lem:masher}
Let $M\in\Znn$ be a nonsingular Smith massager  and $S$ the corresponding Smith form. If $h_i$ is the $i^{th}$ diagonal entry of the lower row Hermite form $H$ of $M$, then $gcd(h_i, s_i)=1$.
\end{lemma}

\begin{proof}
The lemma  follows from the fact that a matrix and its row Hermite form share the same column rank profile. Therefore, since, by Lemma~\ref{lem:pdivs}, the last $i$ columns of $M$ have full rank over $\Z/(p)$ for any $p\mid s_{n-i+1}$, then the last $i$ columns of $H$ have full rank over $\Z/(p)$, and thus, $p \nmid h_{n-i+1}$.
\end{proof}

\begin{lemma} \mylabel{lem:masher2}
Let $M\in\Znn$ be a nonsingular Smith massager  for a matrix $A$, and let $H\in\Znn$ be the lower triangular row Hermite form of $M$. Then, $MH^{-1}$ is a unimodular Smith massager for $A$.
\end{lemma}

\begin{proof}
Since $H$ is unimodularly left equivalent to $M$, we have that matrix $M H^{-1}$ is integral with $\det MH^{-1} = \pm 1$.
It follows that $MH^{-1}$ is unimodular.  It remains to establish
that $MH^{-1}$ is a Smith massager for $A$.
To this end, note that the inverse of any lower triangular matrix
can be decomposed as the product of $n$ pairs of matrices as follows.
\begin{equation}
H^{-1} = \prod_{i=0}^{n-1}
\left[\begin{array}{ccccc}
I &&&& \\
& 1 &&& \\
& -h_{n-i+1,n-i} & 1 && \\
& \vdots && \ddots & \\
& -h_{n,n-i} &&& 1
\end{array}\right]
\left[\begin{array}{ccccc}
I &&&& \\
& 1/h_{n-i} &&& \\
&& 1 && \\
&&& \ddots & \\
&&&& 1
\end{array}\right]
\end{equation}
Thus multiplying $M$ with $H^{-1}$ can be represented as a series
of $n$ products, where each multiplication first applies an operation
of the type as described in Lemma~\ref{lem:smops}(ii), and second
applies one of the type as in Lemma~\ref{lem:smops}(iii) as certified
by Lemma~\ref{lem:masher}. Therefore, $MH^{-1}$ is a
Smith massager for $A$.
\end{proof}

\section{Computational tools} \mylabel{sec:prelim}

An efficient algorithm for computing a Smith massager is given
by \citet{BirmpilisLabahnStorjohann20}. However, this relied on some subroutines
for linear system solving that were restricted to input matrices
$A$ with $2 \perp \det A$.  In this section, we give simple extensions
of these subroutines, enabling us to extend the Smith massager
algorithm of \citet{BirmpilisLabahnStorjohann20} to input matrices
with arbitrary nonzero determinant.

The first subroutine we need is for nonsingular system solving.
Given a nonsingular $A \in \Z^{n \times n}$ and matrix $B \in \Z^{n \times m}$,
together with a lifting modulus $X \in \Z_{>0}$ that
satisfies $X \perp \det A$ and $\log X \in O(\log n + \log \|A\|)$, 
the linear system solving problem is to compute $\Rem(A^{-1}B,X^d)$
for a given precision $d$. The second problem is
integrality certification. Given an $s \in \Z_{>0}$ in addition
to $B$, determine whether $sA^{-1}B$ is integral, and, if so, return
the matrix $\Rem(sA^{-1}B,s)$.
Provided the ``dimension $\times$ precision
$\leq$ invariant'' compromises $m \times  d \in O(n)$ and $m \times (\log \|B\| +\log s)  \in O(n \log X)$
hold, our target complexity for
solving these problems is
\begin{equation} \label{eq:target}
O(n^{\omega}\, \M(\log n + \log \|A\|) \log n)
\end{equation}
bit operations. The algorithm
supporting~\citep[Corollary~7]{BirmpilisLabahnStorjohann19} solves
the first problem in time~(\ref{eq:target}) but was analyzed only
when $X$ is a power of $2$. 
The algorithm for integrality certification
by \cite[Section~2.2]{BirmpilisLabahnStorjohann20} has the same
constraint since it relies on the algorithm supporting~\citep[Corollary~7]{BirmpilisLabahnStorjohann19}.
The analysis in~\citep[Corollary~7]{BirmpilisLabahnStorjohann19} 
exploited the fact that radix conversions to go between the $X$-adic and binary representation of
intermediate integers were not required since $X$ was a power of 2.
Here, we extend the the linear system solving algorithm of
\citet{BirmpilisLabahnStorjohann19} by showing how to choose $X$
to be the power of a small prime.  Even though radix conversions
are now required,
we show how to maintain the cost~(\ref{eq:target}) by keeping intermediate
results in $X$-adic form and only doing radix conversions
at the beginning and end of the process.

Subsection~\ref{ssec:liftp} shows how to choose $X$ as the power
of a small random prime. Subsection~\ref{ssec:double} recalls the
double-plus-one lifting algorithm of \citet{PauderisStorjohann12}
which forms the basis of the linear system solving and integrality
certification algorithms.  Subsections~\ref{ssec:solve}
and~\ref{ssec:intcert} extend the linear system solving and integrality
certification algorithms, respectively, to work with an $X$ as
chosen in Subsection~\ref{ssec:liftp}.
Subsection~\ref{ssec:massager} uses the results developed in
the previous subsections to extend 
the Smith massager algorithm of \citep{BirmpilisLabahnStorjohann20}
to arbitrary nonsingular matrices.

\subsection{Lifting initialization} \label{ssec:liftp}

Let $C$ be an upper bound for $|\det A|$.
\citet[Theorem~18.10]{vonzurGathenGerhard} show how to produce an
integer $p$ the range $6\log C < p < 12\log C$ that is both prime
and satisfies $p \perp \det A$ with probability at least $1/2$. 
If $p$ is prime, we can check if $p \perp \det A$ by trying to compute an $LUP$
decomposition of $A \bmod p$ over $\Z/(p)$. If $p \perp \det A$,
then we can choose our lifting modulus $X$ to be a power of $p$.
In the following lemma, conditions (iii) and (iv) are included
because they are preconditions of the double-plus-one lifting
algorithm described in the next subsection.

\begin{lemma} \mylabel{lem:liftp}
There exists a Las Vegas algorithm
that takes as input a nonsingular $A \in \Z^{n \times n}$, 
and returns as output an odd integer $X$ that satisfies
\begin{itemize}
\item[(i)] $X$ is the power of a prime $p$ with $\log p \in \Theta(\log n
+ \loglog \|A\|)$,
\item[(ii)] $X \perp \det A$,
\item[(iii)] $X \geq \max(10000,3.61n^2\|A\|)$, and
\item[(iv)] $\log X \in O(\log n + \log ||A||)$.
\end{itemize}
The cost of the algorithm is 
$O(n^\omega\, \M(\log n+ \log \|A\|))$ bit operations.
The algorithm returns {\sc FAIL} with probability at most $1/2$.
\end{lemma}

\begin{proof}
By Hadamard's bound we have $C := n^{n/2}\|A\|^n \geq |\det A|$.
By \citet[Theorem~18.10]{vonzurGathenGerhard}, producing an integer $p$
in the range $6 \log C < p < 12 \log C$ that is both prime and does
not divide $\det A$ with probability at least $1/2$ can 
be done within the allotted time.  Proving that $p$ is prime
can be done within the allotted time using the algorithm of~\cite{AgrawalKayalSaxena04}.
If it is determined that $p$ is not prime, then report {\sc Fail}.
Working over $\Z/(p)$, we use $O(n^\omega \, \M(\log p) + n \, \B(\log p))$
bit operations to compute an $LUP$
decomposition~\citep[\S6.4]{AhoHopcroftUllman} of $\Rem(A,p)$.
The $n \, \B(\log p)$ term in this cost estimate is for inverting
the $n$ nonzero pivots arising during the elimination. Computing
$\Rem(A,p)$ and then its $LUP$ decomposition is within our target
cost since $\log p \in O(\log n + \loglog \|A\|)$ and $\B(\log p)
\in O(\M(\log p) (\loglog p))$. If, during the course of the $LUP$
decomposition, it is determined that $A$ is singular modulo $p$,
then return {\sc Fail}.
Otherwise, let $X$ be the smallest power of $p$
which satisfies the third requirement of the lemma. Then, $X$ also
satisfies the fourth requirement.  
\end{proof}

\begin{corollary} \mylabel{cor:liftX}
If $X$ is a lifting modulus as in Lemma~\ref{lem:liftp}, then
$\Rem(A^{-1},X)$ can be computed in time $O(n^{\omega}\, \M(\log n
+ \log \|A\|))$.
\end{corollary}
\begin{proof}
Let $p$ and $LUP$ be as in the proof of Lemma~\ref{lem:liftp}.
Compute $\Rem(A^{-1},p)=
\Rem(P^T U^{-1}L^{-1},p)$, and use $O(\loglog X)$ steps of algebraic
Newton iteration~\citep[Algorithm~9.3]{vonzurGathenGerhard} to lift
$\Rem(A^{-1},p)$ to $\Rem(A^{-1},X)$. The running time is dominated
by the last step of the lifting, which is within the claimed cost.
\end{proof}

\subsection{Double-plus-one lifting} \label{ssec:double}

Let $X$ be a lifting modulus as in Lemma~\ref{lem:liftp}. 
Given a $k \in \Z_{>0}$, the double-plus-one lifting of \citet[Section~3]{PauderisStorjohann12} computes
a straight
line formula that is congruent modulo $X^k$ to the $X$-adic expansion 
\begin{equation} \label{eq:ainv}
A^{-1} \equiv
\ast + \ast X + \ast X^2 + \cdots + \ast X^{k-1} \bmod X^k.
\end{equation}
The straight line formula consists of only $O(\log k)$ 
matrices instead of $k$ as in~(\ref{eq:ainv}). 
More precisely,
given a $k \in \Z_{>0}$ that is one less than a power of $2$, 
double-plus-one lifting 
computes a residue $R \in \Z^{n\times n}$ such that
\begin{equation} A^{-1} = D + A^{-1} R X^k \label{eq:res},
\end{equation}
where $D \in \Z^{n\times n}$ satisfies
$||D|| \leq 0.6X^k$. Note that $D \equiv A^{-1} \bmod X^k$.
Instead of computing $D$ explicitly, double-plus-one
lifting computes a formula
\begin{equation} \label{eq:spainvexp}
D = (\cdots ((\ast(I+\ast X)+\ast X^2)(I + \ast X^3)+
\ast X^6)(I+ \ast X^7) + \ast X^{14})  \cdots),
\end{equation}
where each $\ast$ is an $n \times n$ integer matrix with $\|\ast\| < X$.
The following result is~\citep[Corollary~6]{PauderisStorjohann12}
except that we use Corollary~\ref{cor:liftX} to compute $\Rem(A^{-1},X)$
in the allotted time.

\begin{lemma}{\rm \citep[Corollary~6]{PauderisStorjohann12}}
\mylabel{lem:dpol} 
Assume we have a lifting modulus $X$ as in Lemma~\ref{lem:liftp}.
Let $k \in \Z_{>0}$ be one less
than a power of two. If $\log k \in O(\log n)$, then a residue $R$
as in~(\ref{eq:res}) and a straight line formula for $D$ as shown
in~(\ref{eq:spainvexp}) can be computed in time 
$O(n^\omega\, \M(\log n + \log \|A\|) \log n)$.
\end{lemma}

\subsection{System solving} \label{ssec:solve}

Let $X$ be a lifting modulus as in Lemma~\ref{lem:liftp}. 
Consider equations~(\ref{eq:res}) and~(\ref{eq:spainvexp}). If $k
\geq d$, then given a $B\in\Z^{n\times m}$, we can compute
$\Rem(A^{-1}B,X^d)$ by premultiplying $B$ by the straight line
formula for $D \equiv A^{-1} \bmod X^{k}$ on the right hand side
of~(\ref{eq:spainvexp}), keeping intermediate expressions 
reduced modulo $X^d$. Applying the formula requires doing the following
operation $O(\log k)$ times: premultiplying an $n \times
m$ matrix with entries reduced modulo $X^d$ by an $n \times n$
matrix $\ast$ with $\|\ast\| < X$.
When $X$ is a power of $2$, and $m \times d \in O(n)$,
\citet[Corollary~7]{BirmpilisLabahnStorjohann19} show that this can
be done within our target cost~(\ref{eq:target}). 

When $X$ is not a power of $2$, we need to use radix conversion to
go between the binary and $X$-adic representation of integers. To
avoid unnecessary radix conversions, we can compute the $X$-adic
expansion of $B$ once at the beginning, and then keep intermediate
results in $X$-adic form. The following result is a corollary
of \citet[Theorem~33]{Storjohann04}.

\begin{lemma} \mylabel{lem:ABmult} 
Let $X \in \Z_{>0}$ satisfy $\log X \in O(\log n + \log \|A\|)$.
Let $C \in \Z^{n \times n}$ with $\|C\| < X$ and $B \in \Z^{n \times
m}$ with $B = \Rem(B,X^d)$. If $m\times d \in O(n)$, then $\Rem(CB,
X^d)$ can be computed in time $O(n^\omega\, \M(\log n + \log \|A\|))$,
assuming the input parameter $B$ and output $\Rem(CB,X^d)$ are given
as $X$-adic expansions.
\end{lemma}

The following extends~\citep[Corollary~7]{BirmpilisLabahnStorjohann19}
using Lemmas~\ref{lem:dpol} and~\ref{lem:ABmult}.

\begin{theorem} \mylabel{thm:solve} 
Assume we have a lifting modulus $X$ as in Lemma~\ref{lem:liftp}.
If entries in $B \in \Z^{n \times m}$
are reduced modulo $X^d$ and $m\times d \in O(n)$,
then $\Rem(A^{-1}B,X^d)$ can be computed in time $O(n^{\omega} \,
\M(\log n  +\log \|A\|)\log n)$.
\end{theorem}
\begin{proof}
Using the radix conversion of \cite[Theorem~9.17]{vonzurGathenGerhard},
compute the $X$-adic expansion of $B$ in time $O(nm\, \M(d \log
X) \log d)$.
Simplifying
this cost estimate using $\M(d \log X) \in O(d^{\omega-1}\,
\M(\log X))$ and $d \in O(n/m)$ shows that this is
within the allotted time. Compute a straight line
formula congruent to $A^{-1} \bmod x^d$ using Lemma~\ref{lem:dpol}.
Applying the straight line formula to $B \bmod X^d$
to compute the $X$-adic expansion of $\Rem(A^{-1}B,X^d)$
now requires $O(\log n)$ applications of Lemma~\ref{lem:ABmult},
plus some matrix additions which do not dominate the cost.
Note that the multiplications with
powers of $X$ are free since we are working with
$X$-adic expansions throughout.
Finally, compute $\Rem(A^{-1}B,X^d)$ from its $X$-adic expansion using another radix conversion.
\end{proof}

\subsection{Integrality certification} \label{ssec:intcert}

Any rational number can be written as an integer and a proper
fraction.  For example,
$$
{\frac{9622976468279041913}{21341}}  = 450914974381661 + {\frac{14512}{21341}},
$$
where $450914974381661$ is the quotient and 14512 is the remainder
of the numerator with respect to the denominator.
Similarly, a rational system solution $A^{-1}B$ can have entries with large
numerators compared to denominators.  In some situations only
the information containing  the proper fractional part of the
system solutions is required.
Given an $s \in \Z_{>0}$,
integrality certification can be used to determine whether $sA^{-1}B$
is integral in a cost that depends on $\log \|A\| + \log s + \log \|B\|$ instead of
$\log \|A^{-1}\| + \log s + \log \|B\|$. If $sA^{-1}B$ is integral, the version of
integrality certification developed by
\citet[Section~2.2]{BirmpilisLabahnStorjohann20} also returns the
proper fractional part $\Rem(sA^{-1}B,s)/s$ of $A^{-1}B$, but
required that $2 \perp \det A$.  Using the tools developed in the
previous subsections the algorithm extends easily to handle the
case of an $A$ with arbitrary nonzero determinant.  For completeness,
we give the recipe here.

\begin{enumerate}
\item Using Lemma~\ref{lem:dpol} compute a high-order residue $R\in \Z^{n \times n}$ such that
$A^{-1} = D + A^{-1}R \times X^h$ for an $h\in\Z_{>0}$ such that $X^h > 2sn^{n/2} \|A\|^{n-1} \|B\|$.
\item Using Theorem~\ref{thm:solve}, compute the system solution $\Rem(A^{-1}(sRB), X^{\ell})$ for some $\ell\in\Z_{>0}$
such that $X^{\ell} > 2n\|A\|(0.6sn\|B\|)$.
\item Let $C$ be the matrix that is congruent to $\Rem(A^{-1}(sRB), X^{\ell})$ but with entries
reduced in the symmetric range modulo $ X^{\ell}$.\\
\If $\|C\| < 0.6sn\|B\|$ \Then\\
\ind{1} \Return $\Rem(C \times X^h,s)$\\ 
\Else\\
\ind{1} \Return {\sc NotIntegral}
\end{enumerate}

\begin{theorem}  \mylabel{thm:intcert}
Assume we have a lifting modulus $X$ as in Lemma~\ref{lem:liftp}.
Let $s \in \Z_{>0}$ and $B \in \Z^{n \times m}$ be given. There
exists an algorithm that determines whether $sA^{-1}B$ is integral, and,
if so, returns $\Rem(sA^{-1}B,s)$. If $m \times (\log s + \log
\|B\|) \in O(n\log X)$ and $m\in O(n)$, then the running time is $O(n^{\omega} \,
\M(\log n + \log \|A\|)\log n)$.
\end{theorem}

\subsection{Computing a Smith massager for any $A$} \label{ssec:massager}


Finally, we show how to generalize the Smith massager algorithm
of \citet{BirmpilisLabahnStorjohann20} to arbitrary nonsingular input
matrices by using the results developed in the previous subsections.
We remark that the cost estimate of the following theorem uses $\B$
instead of $\M$ because the algorithm for computing a massager
makes extensive use of gcd computations to compute intermediate
Smith forms.

\begin{theorem} \mylabel{thm:smith}
There exists a Las Vegas algorithm that
takes as input a nonsingular $A \in \Znn$, and returns as output
the Smith form $S\in \Znn$ of $A$ together with
a reduced Smith massager $M\in\Znn$.
The cost of the algorithm is $O(n^\omega\, \B(\log
n + \log \|A\|) (\log n)^2)$ bit operations.
The algorithm returns {\sc FAIL} with probability at most $1/2$.
\end{theorem}

\begin{proof}
\citet[Algorithm {\tt SmithMassager}]{BirmpilisLabahnStorjohann20}
returns a so-called {\em  index-$(0,n)$ Smith massager}.
This is a 4-tuple $(U, M, T, S)$ of
matrices from $\Znn$, such that $T$ is unit upper triangular, $S$
is the Smith form, and the matrix
\begin{equation} \mylabel{eqB}
B = \left [ \begin{array}{cc} A & AM S^{-1} \\
U & (UM + T) S^{-1} \end{array} \right ] \in \Z^{2n\times 2n}
\end{equation}
is unimodular. From (\ref{eqB}) and the fact that $B$ is integral, we have that
\begin{equation} \mylabel{eq:UMT1}
AM\equiv 0 \colmod S
 \mbox{ and } 
UM+T\equiv 0 \colmod S.
\end{equation}
The second equation in (\ref{eq:UMT1}) is equivalent to
\begin{equation} \mylabel{eq:UMT2}
(-T^{-1}U)M\equiv I_n \colmod S,
\end{equation}
implying that the matrix $M$ is a Smith massager for $A$.

To apply \citep[Algorithm {\tt SmithMassager}]{BirmpilisLabahnStorjohann20}
in the case where $A$ may not satisfy $2 \perp \det A$,
we first use the Las Vegas algorithm of Lemma~\ref{lem:liftp} (at most twice)
to compute a lifting modulus $X$ with probability at least $1/4$.
Then we can directly use \citep[Algorithm {\tt
SmithMassager}]{BirmpilisLabahnStorjohann20} but with the following changes:
in the proof of \citep[Theorem~12]{BirmpilisLabahnStorjohann20}
we appeal to Theorem~\ref{thm:intcert} instead of \citep[Theorem~2]{BirmpilisLabahnStorjohann20};
in the proof of \citep[Theorem~21]{BirmpilisLabahnStorjohann20} we appeal
to~Theorem~\ref{thm:solve} instead of \citep[Corollary~7]{BirmpilisLabahnStorjohann19}.
By running this generalization of~\citep[Algorithm {\tt
SmithMassager}]{BirmpilisLabahnStorjohann20} just described (at
most twice) we can compute $S$ and $M$ with probability at least $1/4$.
\end{proof}

By running the Las Vegas algorithm of Theorem~\ref{thm:smith} at
most three times, we obtain the following result, which will be useful
in subsequent sections.

\begin{corollary} \mylabel{cor:smith3}
There exists a Las Vegas algorithm ${\tt SmithMassager}(A)$ 
with the input/output specification and the running time 
stated in Theorem~\ref{thm:smith}.
The algorithm returns {\tt FAIL} with probability at most $1/8$.
\end{corollary}

\section{Partial linearization} \label{sec:lin}

The cost of algorithms that take as input an integer matrix
$A \in \Z^{n \times m}$ are typically expressed in terms of
the dimensions $n$ and $m$, and  $\log \|A\|$,
which is proportional to
the bitlength of the largest entry of $A$ in absolute value.
More precisely, let us define $\length(a)$ for an integer $a$
to be the number of bits in its binary representation, that is,
$$
\length(a) := \left \{ \begin{array}{ll} 1 & \hbox{if $a=0$} \\
                                  1 +  \lfloor \log_2 |a| \rfloor
    & \hbox{otherwise}
\end{array} \right. .
$$
By extension, for a matrix we define $\length(A) := \length(||A||)$, so
$\length(A)$ is the length of the largest entry of $A$ in absolute
value.

But consider decomposing $A$ into columns as 
$$
A = \left [ \begin{array}{c|c|c} v_1 & \cdots & v_m \end{array}
\right ] \in \Z^{n \times m}.
$$ 
For some inputs, the lengths of the columns $v_i$ can be skewed, that is, the \emph{average} column length 
$$
d =  \left\lceil \sum_{i=1}^m \length(v_i)/m\right\rceil
$$
can be asymptomatically smaller than $\length(A) = \max_i \length(v_i)$.
Even $\length(A) \approx md$ is possible in the case of one column
of large length. For such inputs, being able to replace the term
$\length(A)$ with the average length $d$ can give significantly
improved cost estimates.

\begin{example}
For the identity matrix $I_m$, we have $\length(A)=1$ and
the average column length is also
$d=1$. Now let $I_m'$ be equal to $I_m$ but with the last
column multiplied by $2^{m+1}-1$. Then $\length(I')=m+1$ but the
average column length is only $d=2$.
\end{example}

In this section, we adapt the partial linearization technique for polynomial matrices given by
 \citet[Section~6]{GuptaSarkarStorjohannValeriote11} to the case of integer matrices.
The main motivation is to extend the algorithms from Section~\ref{sec:prelim} so that
their cost estimates depend on the average $\length$ $d$ and not $\length(A)$.

The technique transforms the input matrix $A$ into a new matrix $D$
which can be used in place of $A$ for all of the algorithms presented
in Section~\ref{sec:prelim}, and many more (see below and also the remarks at the end of Subsection~\ref{ssec:permbnd}). 
Matrix $D$ will satisfy
that $\length(D) \leq d+1$, at the cost of $D$ having at most $m$ more rows
and columns than $A\in\Z^{n\times m}$.

More importantly, the constructed matrix $D$ will ``imitate" $A$ in a way such that the output of the routines with $D$ as input includes the original output in a direct way. Specifically, matrix $D$ will satisfy the following two fundamental properties with respect to $A$:
\begin{itemize}
\item[(i)] $D$ can be obtained from $\diag(A, I)$ using unimodular row and column operations.
\item[(ii)] The principal $n\times n$ submatrix of the adjoint of $D$ equals the adjoint of $A$ (for square matrices).
\end{itemize}

Property (i) establishes that the rank, the determinant (for square
matrices) and the Smith form of matrix $A$ can be trivially deduced
from the same objects for matrix $D$. In Subsection~\ref{ssec:maslin}
we show that computing the Smith massager of a nonsingular $A$
can also be directly reduced to computing the Smith massager of $D$.

Property (ii) provides us with a direct extension of system solving. 
If $A \in \Z^{n \times n}$ is nonsingular, then for any
matrix $B\in\Z^{n\times \ast}$, we have that the first $n$ rows of
$$D^{-1}\left[\begin{array}{c} B \\ 0 \end{array}\right]$$
are equal to $A^{-1}B$.
Finally, because $\det D=\det A$ and using property (ii),
it follows that the principal $n\times n$ submatrix of the lower
row Hermite form of $D$ equals the lower row Hermite form of $A$.

\begin{example} \label{ex:one}
Let 
\[
A = \left[ \begin {array}{cccc} 2&4&44199&3061969404\\ 4&8&19644&765492351\\ 7&8&44199&5358446457\\ 7&5&9822&765492351\end {array} \right] \in\Z^{4\times 4},
\]
a matrix with skewed column lengths.  In this case $\length(A)=33$ and average column length is $d=14$. The partial linearization of $A$ constructed later 
in this section will be
\[D = \left[ \begin {array}{ccccccc} 2&4&11431&12796&2&6663&11\\
4&8&3260&15487&1&13953&2\\ 
7&8&11431&10105&2&15757&19\\
7&5&9822&15487&0&13953&2\\
&&-16384&&1&&\\ 
&&&-16384&&1&\\
&&&&&-16384&1\end {array} \right] \in \Z^{7\times 7}.\]
Notice that $\|D\|\leq 2^d=16384$.
\end{example}

\subsection{The partial linearization construction}

Let $e \in \Z_{\geq 0}$ and $d \in
\Z_{\geq 1}$ be given and assume for the moment that
a column vector $v \in \Z_{\geq 0}^{n \times 1}$
contains only nonnegative entries. Then, we define
$C_{e,d}(v)$ to be the unique $n \times e$ matrix over $\Z_{\geq 0}$
that satisfies
\begin{equation} \label{eq:quoC}
\Quo(v,2^d) 
= C_{e,d}(v) \left [ \begin{array}{c} 1 \\ 2^d \\ \vdots \\ 2^{(e-1)d}
\end{array} \right ],
\end{equation}
with all but possibly the last column (if $e>0$) of magnitude strictly
less than $2^d$.
If $e=0$ then $C_{e,d}(v)$ is the $n \times 0$ matrix, while for $e\geq 1$,
\begin{equation} \label{eq:xadic} v  = \Rem(v,2^d) + 
\col(C_{e,d}(v),1)2^d+ \cdots + \col(C_{e,d}(v),e)2^{ed}\end{equation}
is the
$2^d$-adic series expansion of $v$, except that the coefficient
$\col(C_{e,d}(v),e)$ of $2^{ed}$ may have magnitude
greater than or equal to $2^d$.

\begin{example}
For $v = \left [ \begin{array}{c} 29821\end{array} \right ]$,
$\Rem(v, 2^3) = 5$ and $C_{3,3}(v) = \left [ \begin{array}{ccc} 
7 & 1 & 58 \end{array} \right ].$
\end{example}

We can extend the definition of
$C_{e,d}$ to an arbitrary vector $v\in \Z^{n \times 1}$ in the following way. Let $v^{(+)}$ denote the vector $v$ but with all negative entries zeroed out, and  $v^{(-)} := v - v^{(+)}$ denote the vector $v$ but with all but the positive entries zeroed out.  Then, $v^{(+)}$ and $-v^{(-)}$ are over
$\Z_{\geq 0}$, and $v = v^{(+)}-(-v^{(-)})$. Finally we let
\[\CedStar_{e,d}(v) :=  C_{e,d}(v^{(+)}) - C_{e,d}(-v^{(-)}),\]
which still satisfies equations (\ref{eq:quoC}) and (\ref{eq:xadic}) if we  replace $\Rem$ and $\Quo$ by
\[\RemStar(v, 2^d) := \Rem(v^{(+)}, 2^d) - \Rem(-v^{(-)}, 2^d),\]
\[\QuoStar(v, 2^d) := \Quo(v^{(+)}, 2^d) - \Quo(-v^{(-)}, 2^d).\]

We define structured matrices $E_d$ and $F_d$ by
$$
E_d := -2^d\,\col(I,1) = 
\left [ \begin{array}{c} -2^d \\ \\ \\ \\ \\
\end{array} \right]\hbox{~~and~~} F_d := \left [ \begin{array}{ccccc}
1 &      & & & \\
-2^d & 1 &      & \\
 & -2^d & \ddots &        \\
 & & \ddots &  1    &      \\
 & & &  -2^d  & 1  \\
\end{array} \right ],
$$  
with the dimensions of $E_d$ and $F_d$ to be determined by the context.
We remark that $F_d^{-1}$ will be the unit lower triangular Toeplitz
matrix with $2^{id}$ on the $i$th subdiagonal. The next lemma follows from the definition of $E_d$ and $F_d$ and equations (\ref{eq:quoC}) and (\ref{eq:xadic}).

\begin{lemma} \label{lem:plin1}
Given $v\in\Z^{n\times 1}$, $e \in \Z_{\geq 0}$ and $d \in \Z_{\geq 1}$, let
\[c := \left\{\begin{array}{cc} v & \text{if } e=0 \\ \RemStar(v, 2^d) & \text{if } e>0 \end{array}\right.,\]
and
\[Q_{e,d}(v) = \left[\begin{array}{c|c|c} \QuoStar(v, 2^d) & \cdots & \QuoStar(v, 2^{ed}) \end{array}\right].\]
Then,
\begin{equation}
\left[\begin{array}{c|c}
c & \CedStar_{e,d}(v) \\\hline
E_d & F_d
\end{array}\right] =
\left[\begin{array}{c|c}
I_n & Q_{e,d}(v) \\\hline
 & I_e
\end{array}\right]
\left[\begin{array}{c|c}
v & \\\hline
 & I_e
\end{array}\right]
\left[\begin{array}{c|c}
1 & \\\hline
E_d & F_d
\end{array}\right].
\end{equation}
\end{lemma}

By replacing the single column vector $v$ with a matrix $A = \left [ \begin{array}{c|c|c} v_1 & \cdots & v_m \end{array} \right ]$ of $m$ column vectors $v_i$, we obtain:
\begin{corollary} \mylabel{cor:Ddef}
Given $A = \left [ \begin{array}{c|c|c} v_1 & \cdots & v_m \end{array} \right ] \in \Z^{n \times m}$, $\bar{e} = (e_1,\ldots, e_m) \in \Z^m_{\geq 0}$ and $d \in \Z_{\geq 1}$. Let
\[c_i := \left\{\begin{array}{cc} v_i & \text{if } e_i=0 \\ \RemStar(v_i, 2^d) & \text{if } e_i>0 \end{array}\right.,\]
for $1\leq i\leq m$, and define the matrix
$$D = D_{\bar{e},d}(A) := 
\left [ \begin{array}{c|c|c||c|c|c}
c_1 & \cdots & c_m & \CedStar_{e_1,d}(v_1)
& \cdots & \CedStar_{e_m,d}(v_m) \\\hline \hline
E_d &  &  &  F_d  &  &  \\ \hline
    & \ddots  &  & & \ddots & \\ \hline
  &       &  E_d  & &  & F_d
\end{array} \right ] \in \Z^{\bar{n} \times \bar{m}},$$
with $\bar{n} = n + e_{[m]}$ and $\bar{m} = m + e_{[m]}$, where $e_{[m]}=e_1+\cdots+e_m$. Then, matrix $D$ satisfies
\begin{equation} \mylabel{eq:Ddecomp}
D =
\left[\begin{array}{cc}
I_n & Q \\
 & I_{e_{[m]}}
\end{array}\right]
\left[\begin{array}{cc}
A & \\
 & I_{e_{[m]}}
\end{array}\right]
\left[\begin{array}{cc}
I_m & \\
E & F
\end{array}\right],
\end{equation}
where $Q = \left[\begin{array}{c|c|c} Q_{e_1,d}(v_1) & \cdots & Q_{e_m,d}(v_m) \end{array}\right]\in\Z^{n\times e_{[m]}}$, $E=\diag(E_d, \ldots, E_d)\in\Z^{e_{[m]}\times m}$ and $F=\diag(F_d, \ldots, F_d)\in\Z^{e_{[m]}\times e_{[m]}}$.
\end{corollary}
From equation (\ref{eq:Ddecomp}), it is apparent that $D$ enjoys the following properties:
\begin{corollary} \mylabel{cor:lin2}
Given $A  \in \Z^{n \times m}$,
 $\bar{e} = (e_1,\ldots, e_m) \in \Z^m_{\geq 0}$ and $d \in \Z_{\geq 1}$. Let $D = D_{\bar{e},d}(A)$ as in Corollary~\ref{cor:Ddef}. Then
\begin{itemize}
\item[(i)] ${\rm rank}(D) = {\rm rank}(A) + e_{[m]}$.
\item[(ii)]  $D$ has the same Smith form as $A$ up to additional trivial invariant factors.
\end{itemize}
Furthermore, if $n=m$, then:
\begin{itemize}
\item[(iii)] $\det D = \det A$.
\item[(iv)] The principal $n\times n$ submatrix of the adjoint of $D$ equals the adjoint of $A$.
\end{itemize}
\end{corollary}

Notice that Corollary~\ref{cor:Ddef} does not make any assumptions
on the parameters $\bar{e}$ and $d$. The properties of matrix $D =
D_{\bar{e},d}(A)$ corresponding to the original matrix $A$ are true
for any $\bar{e}$ and $d$. However, the partial linearization
technique is particularly useful if we pick $\bar{e}$ and $d$ in a
way such that $\bar{m} \in O(m)$ and $\log \|D\|$ corresponds to the
the average $\length$ of the columns of $A$.
The following is the main result of this section.

\begin{theorem} \mylabel{thm:lin}
Given matrix $A = \left [ \begin{array}{c|c|c} v_1 & \cdots & v_m \end{array} \right ] \in \Z^{n \times m}$, let
\[ d := \left\lceil \sum_{i=1}^m \length(v_i)/m\right\rceil, \]
$\bar{e} = (e_1,\ldots, e_m) \in \Z^m_{\geq 0}$ where each $e_i \in \Z_{\geq 0}$ is chosen minimal such that $\length(v_i) \leq (e_i+1)d$, and $D = D_{\bar{e},d}(A)$. Then:
\begin{itemize}
\item $||D|| \leq 2^d$,
\item $\bar{n} < n + m$ and $\bar{m} < 2m$.
\end{itemize}
\end{theorem}

\begin{proof}
The choice of $e_i$ ensures that, for each
$v_i$, the expansion  in (\ref{eq:xadic}) is the $2^d$-adic
expansion of $v$.  This shows that the length of all entries in the
first $n$ rows of $D$ are bounded by $d$. Since the entries in the
last $\bar{n}-n$ rows of $D$ are bounded in magnitude by $2^d$, the
claimed bound for $||D||$ follows.

To prove our upper bounds for $\bar{n}$ and $\bar{m}$ we
show that $\sum_{i=1}^m e_i < m$.
Note that $e_i$ is precisely defined as
\[e_i = \left\lceil\frac{\length(v_i)}{d}-1\right\rceil < \frac{\length(v_i)}{d},\]
and so 
\[\sum_{i=1}^m e_i < \sum_{i=1}^m \frac{\length(v_i)}{d} \leq m.\]
\end{proof}

\begin{example} \label{ex:plin1}
Let 
\[A = \left[ \begin {array}{cccc} 2&4&44199&3061969404\\ 4&8&19644&765492351\\ 7&8&44199&5358446457\\ 7&5&9822&765492351\end {array} \right], 
\]
be the matrix from Example \ref{ex:one}.
Then, with the average (column) length $d=14$ and $\bar{e}=(0, 0, 1, 2)$ we get
\[D = \left[ \begin {array}{cccc||ccc} 
2&4&11431&12796&2&6663&11\\
4&8&3260&15487&1&13953&2\\ 
7&8&11431&10105&2&15757&19\\
7&5&9822&15487&0&13953&2\\ \hline \hline 
&&-16384&&1&&\\
&&&-16384&&1&\\ 
&&&&&-16384&1\end {array} \right].\]
One can easily verify that the adjoint of $A$ lies in the principal $4 \times 4$ sub-matrix of the adjoint of $D$, and that the Smith form of $A$ lies in the trailing $4 \times 4$ sub-matrix of the Smith form of $D$.
\end{example}

The approach of Corollary~\ref{cor:Ddef} can also be used to partially
linearize the rows of a matrix $A$. If we transpose a
matrix $A$ with skewed row lengths, then it has skewed column
lengths. Then, by transposing the linearization of $A^T$, it satisfies
all the properties given in Corollary~\ref{cor:lin2}. We can see
that from the row linearization equivalent of equation (\ref{eq:Ddecomp}),
which is
\begin{equation} \mylabel{eq:Ddecomp2}
D_{\bar{e},d}(A^T)^T=\left[\begin{array}{cc} I & E^T \\ & B^T \end{array}\right]
\left[\begin{array}{cc} A & \\ & I \end{array}\right]
\left[\begin{array}{cc} I & \\ Q^T & I \end{array}\right].
\end{equation}
\begin{corollary} \mylabel{cor:rowlin}
Let $A \in \Z^{m \times n}$, and consider the matrix $D = D_{\bar{e},d}(A^T)^T$. The magnitude of the entries in $D$ will then be bounded by $2^d$ where $d$ is the average $\length$ over the rows of $A$, and $D$ will enjoy all the properties following from Corollary~\ref{cor:Ddef} and Theorem~\ref{thm:lin}.
\end{corollary}

\subsection{The permutation bound} \mylabel{ssec:permbnd}

Our approach so far is particularly effective for matrices $A\in\Znn$
where the average of the sum of the lengths of the columns (or rows) is small
compared to $\length(A)$.  However, the technique is not useful for
input matrices that have, simultaneously, some columns and rows of
large length. For this reason, as in the case of polynomial
matrices~\citep[Section~6]{GuptaSarkarStorjohannValeriote11}, we
develop an approach to handle such inputs based on the following
{\em a priori} upper bound for $|\det A|$.

By definition, $\det A = \sum_{\sigma\in S_n} \sign(\sigma) \prod_{i=1}^n A_{i,\sigma_i}$, where $S_n$ is the set of all permutations of $(1, 2, \ldots, n)$. Therefore,
\[\det A \leq n!\max_{\sigma\in S_n} \prod_{i=1}^n |A_{i,\sigma_i}|,\]
and so, we define
\[{\rm PermutBnd}(A) := \max_{\sigma\in S_n} \sum_{i=1}^n \length(A_{i,\sigma_i}).\]

As in the polynomial case, up to a row and column permutation, we may assume that $d_i := \length(A_{i,i})$ bounds the length of the submatrix $A_{i\ldots n,i\ldots n}$, for $1 \leq i \leq n$.
Such a row and column permutation can be found by sorting the set of triples $\{(i, j, |A_{i,j}|)\}_{1\leq i,j\leq n}$ into nonincreasing order according to their third component. Then, by definition, $d_1+\cdots+d_n\leq {\rm PermutBnd}(A)$.

Let $d:=\left\lceil\sum_{i=1}^n d_i/n\right\rceil$ and $\bar{e}=(e_1,\ldots,e_n)$ with $e_i\in\Z_{\geq 0}$ minimal such that $d_i\leq (e_i+1)d$. Then, due to the choice of $d_i$, row $i$ of matrix $D_{\bar{e},d}(A)$ will have length bounded by $d_i+1$ for $1\leq i\leq n$, and all the extra rows will have length bounded by $d+1$. Furthermore, let $\bar{e}'$ contain $\bar{e}$ augmented with $\sum_{i=1}^n e_i$ zeros. We have the following corollary for matrix $D:=D_{\bar{e}',d}(D_{\bar{e},d}(A)^T)^T$.

\begin{corollary} \mylabel{cor:cor9}
Let $A\in\Znn$ be given. Using the choices for $d$, $\bar{e}$ and $\bar{e}'$ as specified above, the matrix $D:=D_{\bar{e}',d}(D_{\bar{e},d}(A)^T)^T\in\Z^{\bar{n}'\times \bar{n}'}$ satisfies
\begin{itemize}
\item[(i)] $\|D\|\leq 2^d$ with $d\leq \lceil{\rm PermutBnd}(A)/n\rceil$, and
\item[(ii)]  $\bar{n}'< 3n$,
\end{itemize}
along with all the properties from Corollary~\ref{cor:lin2}.
\end{corollary}

\begin{remark}[Application to system solving] The fact that the
principal $n \times n$ submatrix of the adjoint of  the partially
linearized matrix $D$ is equal to the adjoint of the original matrix
$A$ provides us with a direct extension to system solving.  For any
matrix $B\in\Z^{n\times m}$, we have that the first $n$ rows of
$$D^{-1}\left[\begin{array}{c} B \\ 0 \end{array}\right]$$ are equal to
$A^{-1}B$. Therefore, Theorem~\ref{thm:solve} can have cost which
depends on the average bitlength $d$ of $A$ and not the bitlength
of the largest entry. The average bitlength $d$ can assume any of
the three definitions given by Theorem~\ref{thm:lin},
Corollary~\ref{cor:rowlin} and Corollary~\ref{cor:cor9}.
\end{remark}

\begin{remark}[Application to integrality certification]
Suppose $D$ is a partial linearization of $A$.
For any $s \in \Z_{>0}$ and $B \in \Z^{n \times m}$, 
it follows from equations~(\ref{eq:Ddecomp}) and~(\ref{eq:Ddecomp2})
that $$sD^{-1}\left[\begin{array}{c} B \\ 0 \end{array}\right]$$
will be integral if and only if $sA^{-1}B$ is integral.  Therefore,
Theorem~\ref{thm:intcert} can have cost which depends on the average
bitlength $d$ of $A$ and not the bitlength of the largest entry.
The average bitlength $d$ can assume any of the three definitions
given by Theorem~\ref{thm:lin}, Corollary~\ref{cor:rowlin} and
Corollary~\ref{cor:cor9}.
\end{remark}

\begin{remark}[Application to inverting unimodular matrices]
Suppose $D$ is a partial linearization of a unimodular matrix $A$.
A straight line formula for $A^{-1}$ is given by $$\left [
\begin{array}{c|c} I_n & 0 \end{array} \right ] T \left [
\begin{array}{c} I_n \\\hline 0 \end{array} \right ]$$ where $T$
is a straight line formula for the inverse of a partial
linearization of $A$. 
Such a straight line formula for $A^{-1}$ can thus be computed deterministically
in $O(n^{\omega}\, \M(\log n + d) \log n)$ bit
operations by \cite[Section~3]{PauderisStorjohann12},
where $d$ is the average bitlength of $A$ according to any of
the three definitions given by Theorem~\ref{thm:lin},
Corollary~\ref{cor:rowlin} and Corollary~\ref{cor:cor9}.
\end{remark}

\begin{remark}[Application to computing the Hermite form]
If $A\in \Z^{n \times n}$ 
is nonsingular, then the lower triangular row Hermite form of
$A$ shows up as the principal $n \times n$ submatrix of the Hermite
form of the partially linearized matrix $D$.
\end{remark}

\begin{example}
The lower triangular row Hermite form of the matrix $D$ from Example \ref{ex:plin1} is 
$$
\left[ \begin {array}{cccc||ccc}
777&&&&&&\\
401&1&&&&&\\
174&0&4911&&&&\\
762&0&0&765492351&&&\\ \hline \hline
696&0&3260&0&1&&\\
762&0&0&765475967&&1&\\
762&0&0&497056895&&&1
\end {array} \right] 
$$
with the $4 \times 4$ principal sub-matrix being the corresponding lower triangular row Hermite form of $A$.
\end{example}

\subsection{Smith massagers and  partial linearization} \mylabel{ssec:maslin}

We can also employ the partial linearization technique to replace the $\log\|A\|$ term in Theorem~\ref{thm:smith} with the average bitlength $d$ of the columns (or rows) in $A$.

\begin{theorem} \mylabel{thm:smlin}
Let $A\in\Znn$ and $D\in\Z^{\bar{n}\times \bar{n}}$ be the partially linearized version of $A$ from Theorem~\ref{thm:lin}. If
\begin{equation}  \mylabel{eq:smlarge}
\left(\left[\begin{array}{cc} I_{\bar{n}-n} & \\ & S \end{array}\right], \left[\begin{array}{cc} 0 & M_1 \\ 0 & M_2 \end{array}\right]\right)\end{equation}
is a Smith massager for $D$, where $S\in\Znn$, $M_1\in\Znn$ and $M_2\in\Z^{(\bar{n}-n)\times n}$, then $(S, M_1)$ is a Smith massager for $A$.
\end{theorem}

\begin{proof}
We will show that $S, M_1\in\Znn$ satisfy Definition~\ref{def:SM} for $A$.

From Theorem~\ref{thm:lin}, we have that
\begin{align*}
D\left[\begin{array}{cc} 0 & M_1 \\ 0 & M_2 \end{array}\right]
&= \left[\begin{array}{cc} I_n & Q \\ & I_{\bar{n}-n} \end{array}\right]
\left[\begin{array}{cc} A & \\ & I_{\bar{n}-n} \end{array}\right]
\left[\begin{array}{cc} I_n & \\ E & F \end{array}\right]
\left[\begin{array}{cc} 0 & M_1 \\ 0 & M_2 \end{array}\right] \\
&= \left[\begin{array}{cc} I_n & Q \\ & I_{\bar{n}-n} \end{array}\right]
\left[\begin{array}{cc} 0 & AM_1 \\ 0 & EM_1+FM_2 \end{array}\right].
\end{align*}
Since~(\ref{eq:smlarge}) is a Smith massager for $D$, it follows from 
Definition~\ref{def:SM}.(i) that
\[D\left[\begin{array}{cc} 0 & M_1 \\ 0 & M_2 \end{array}\right] \equiv 0 \colmod \left[\begin{array}{cc} I_{\bar{n}-n} & \\ & S \end{array}\right],\]
it follows that
\[\left[\begin{array}{cc} 0 & AM_1 \\ 0 & EM_1+FM_2 \end{array}\right] \equiv 0 \colmod \left[\begin{array}{cc} I_{\bar{n}-n} & \\ & S \end{array}\right],\]
and that
\[AM_1\equiv 0 \colmod S.\]
Moreover, since $B$ is unit lower triangular, we see that
\[M_2\equiv -F^{-1}EM_1 \colmod S.\]
Finally, by Definition~\ref{def:SM}.(ii), there exist a matrix $W_D\in\Z^{\bar{n}\times \bar{n}}$ such that
\[W_D\left[\begin{array}{cc} 0 & M_1 \\ 0 & M_2 \end{array}\right]\equiv \left[\begin{array}{cc} I_{\bar{n}-n} & \\ & I_n \end{array}\right] \colmod \left[\begin{array}{cc} I_{\bar{n}-n} & \\ & S \end{array}\right].\]
The last equation can be transformed to
\[\left( W_D \left[\begin{array}{cc} I_n & \\ -F^{-1}E & I_{\bar{n}-n} \end{array}\right]\right) \left[\begin{array}{cc} 0 & M_1 \\ 0 & 0 \end{array}\right]\equiv \left[\begin{array}{cc} I_{\bar{n}-n} & \\ & I_n \end{array}\right] \colmod \left[\begin{array}{cc} I_{\bar{n}-n} & \\ & S \end{array}\right],\]
from which it directly follows that there exists a matrix $W\in\Znn$ such that
\[WM_1\equiv I_n\colmod S.\]
\end{proof}

Furthermore, by equation (\ref{eq:Ddecomp2}) and by following the same steps as in the proof Theorem~\ref{thm:smlin}, we obtain the following corollary.

\begin{corollary} \mylabel{cor:smlin2}
Let $A\in\Znn$ and $D\in\Z^{\bar{n}\times \bar{n}}$ be the partially linearized version of $A$ from Corollary~\ref{cor:rowlin} or Corollary~\ref{cor:cor9}. If
\[\left(\left[\begin{array}{cc} I_{\bar{n}-n} & \\ & S \end{array}\right], \left[\begin{array}{cc} 0 & M_1 \\ 0 & M_2 \end{array}\right]\right)\]
is a Smith massager for $D$, where $S\in\Znn$, $M_1\in\Znn$ and $M_2\in\Z^{(\bar{n}-n)\times n}$, then $(S, M_1)$ is a Smith massager for $A$.
\end{corollary}

\section{Example}\label{sec:informal}

In this section, we illustrate our Smith form with multipliers
algorithm using the following example. We have already discussed
the algorithm in Section~\ref{ssec:appruni}, and we will rigorously
establish it in Sections~\ref{sec:random}--\ref{sec:algo}.

\begin{example}
Let our input matrix be
\[
 A := \left[ \begin {array}{ccccccc} 1&0&0&0&0&0&0\\ 1&1&1&1&1&1&1\\ 1&2&4&1&2&4&1\\ 1&3&2&6&4&5&1\\ 1&4&2&1&4&2&1\\ 1&5&4&6&2&3&1\\ 1&6&1&6&1&6&1\end {array} \right].
 \]
Given as input $2A$, the algorithm supporting Theorem~\ref{thm:smith}
returns the Smith form $2 S$ and a Smith massager  $M$ for $2A$:
 \[
 2 S := \left[ \begin {array}{ccccccc} 2&&&&&&\\ &2&&&&&\\ &&2&&&&\\ &&&2&&&\\ &&&&2&&\\ &&&&&16&\\ &&&&&&160\end {array} \right],~~
 M := \left[ \begin {array}{ccccccc} 
 1&0&1&1&2&8&0\\ 
 0&1&1&0&2&11&65\\ 
 1&0&1&1&1&12&15\\ 
 0&1&1&1&3&6&98\\ 
 0&0&0&0&0&12&155\\ 
 1&1&1&1&1&7&125\\ 
 1&1&1&1&1&0&2
 \end {array}\right]
. \]
We always take $M$ to be reduced column modulo $2S$, that is, it
should be a reduced Smith massager.

The next step is to pick a random matrix
\[
R := \left[ \begin {array}{ccccccc}
0&0&1&1&1&0&1\\
1&1&0&0&1&0&1\\
0&0&0&0&0&0&0\\
1&1&1&1&0&0&0\\
0&1&0&1&0&1&0\\
0&1&1&1&1&1&0\\
1&1&0&0&1&1&0
\end {array} \right],
\]
where each entry is chosen independently and uniformly from a set $[0,\lambda-1]$ of $\lambda\in O(n\|A\|)$ consecutive 
integers. (For the example, we let $\lambda := 2$.)

By perturbing $M$ by the random choice of $R$ post-multiplied with $2 S$, we obtain
\[
B := M+2 RS = \left[ \begin {array}{ccccccc} 
1&0&3&3&6&8&160\\ 
2&3&1&0&6&11&225\\ 
1&0&1&1&1&12&15\\ 
2&3&3&3&3&6&98\\ 
0&2&0&2&0&28&155\\ 
1&3&3&3&5&23&125\\ 
3&3&1&1&5&16&22\end {array} \right] ,
\]
which, by Proposition~\ref{prp:cARS}, is a Smith massager for $A$.

Computing the lower triangular row Hermite form of the random matrix $B$, gives
\[
H := \left[ \begin {array}{ccccccc}
 830295&&&&&&\\ 
 547348&1&&&&&\\ 
 602711&&1&&&&\\ 
 592450&&&1&&&\\ 
540934&&&&1&&\\ 
350043&&&&&1&\\ 
323815&&&&&&1\end {array} \right].
\]
Our aim is for $H$ to have only the first diagonal entry non-trivial.
If $B$ is not left equivalent to such a matrix $H$, then the algorithm
fails. This happens, for example, if the
random $R$ has the entry in row $1$ and column $6$ equal to 
$1$ rather than $0$.  Showing that
the Hermite form of $B$ is almost trivial with high probability
is the main focus of Section~\ref{sec:random}. Then, in
Section~\ref{sec:herm}, we give an algorithm to assay if the Hermite
form of $B$ has the desired structure, and if so, to compute the Hermite form itself.

To obtain a unimodular Smith massager, we simply extract $H$ from $B$ by post-multiplying with $H^{-1}$.
\[
V := BH^{-1} = \left[ \begin {array}{ccccccc} 
-74&0&3&3&6&8&160\\
-99&3&1&0&6&11&225\\ 
-13&0&1&1&1&12&15\\ 
-49&3&3&3&3&6&98\\ 
-75&2&0&2&0&28&155\\ 
-68&3&3&3&5&23&125\\ 
-22&3&1&1&5&16&22
\end {array} \right].
\]
By construction, the matrix $V$ is integral and unimodular. In addition, and as proven by Lemma~\ref{lem:masher2}, $V$ is a Smith massager for $A$.

The fact that $H$ has only one non-trivial column allows us to
easily establish a nice bound on the size of matrix $V$. Notice
that the columns of $V$ have the same bitlength as the columns of
$B$ except for only the first column. In addition, the bitlength
of the columns of $B$ equals the bitlength of the columns of the
Smith massager $M$ plus the bitlength of $\lambda$. In
Section~\ref{sec:algo}, we give the overall algorithm for computing
the Smith multipliers and establish explicit bounds on the size of
their entries.

Finally, since $V$ is a unimodular Smith massager for $A$, this makes the matrix
\[
U := AVS^{-1} =  \left[ \begin {array}{ccccccc} 
-74&0&3&3&3&1&2\\ 
-400&14&12&13&13&13&10\\ 
-817&28&25&27&25&31&20\\ 
-1353&53&42&47&37&43&34\\ 
-1003&32&19&23&25&32&26\\ 
-1291&49&40&39&39&36&33\\ 
-1480&59&47&43&48&38&38
\end {array} \right] 
\]
also integral and unimodular. By construction, the two unimodular matrices $V,U\in\Znn$ satisfy $AV=US$.
\end{example}

\section{Random perturbations of Smith massagers}\mylabel{sec:random}

Let $A \in \Z^{n\times n}$ be nonsingular with Smith form $S.$ In
this section, we show how to perturb a Smith massager $M$ for $A$
into a unimodular Smith massager $V$. The first step will be to
obtain a Smith massager $B := M + RS$ that is left equivalent  (over
$\Z$) to a lower triangular row Hermite form with the shape
\begin{equation} \mylabel{niceherm}
\left [ \begin{array}{ccccc}
|\det B| &  & & & \\
\ast & 1 & & & \\
\ast & & 1 & \\
 \vdots & & & \ddots & \\
\ast & & & & 1 
\end{array} \right ] \in \Z^{n \times n}.
\end{equation}
The property that the last $n-1$ diagonal entries of $B$ are equal
to $1$ coincides with the property that the last $n-1$ columns of
$B \bmod p$ are linearly independent over $\Z/(p)$ for all primes
$p$.

Our approach is inspired by and follows that of \citet[Section~6]{EberlyGiesbrechtVillard}, where the following
general result is established: for $\lambda \geq 2$, a matrix $R
\in \Z^{n \times n}$ with entries chosen uniformly and randomly
from $[0,\lambda-1]$ will have an expected number of $O(\log_{\lambda}n)$
nontrivial invariant factors.


\begin{theorem} \mylabel{thm:cond}
Let $A \in \Z^{n\times n}$ be nonsingular with Smith form $S$. Let
$M$ be a reduced Smith massager for $2A$.
For any $R \in \Z^{n \times n}$,
\begin{itemize}
\item[(i)] the matrix $B := M + 2RS$ is a Smith massager for $A$, and
\item[(ii)]  if entries in $R$ are chosen
chosen uniformly and randomly from
$[0,\lambda-1]$, where
$$\lambda = 105\max(n,\left\lceil(\det 2S)^{1/n}\right\rceil)
,$$
then the probability that there exists a prime $p$ such that
the last $n-1$ columns of $B\bmod p$ are linearly dependent
over $\Z/(p)$ is less than $1/2$.
\end{itemize}
\end{theorem}

Part (i) of Theorem~\ref{thm:cond} follows directly from Proposition~\ref{prp:cARS}, so it remains only to prove part (ii).
This will be done using a sequence of lemmas. 
For the rest of this section, we let $A$, $S$, $M$, $R$, $\lambda$ and
$B=M+2RS$ be as defined in Theorem~\ref{thm:cond}.

We start by defining a set of probabilistic events that will
facilitate the proofs in this section.  For a prime $p$ and $1 \leq
m \leq n-1$, let $\dep_m^p$ denote the event that the last $m$
columns of $B$ are linearly dependent modulo $p$. To complete the
proof of Theorem~\ref{thm:cond} we show that ${\rm Pr}[\vee_p
\dep_{n-1}^p] < 0.5$, where $\vee_p$ means ranging over all primes.
We begin with Lemmas~\ref{lem:small} and~\ref{lem:cprob} that hold
for all primes $p$.  Then, following
\citet[Section~6]{EberlyGiesbrechtVillard}, we will separately
consider the small primes $p < \lambda$ in Subsection~\ref{ssec:smallprimes},
and the large primes $p \geq \lambda$ in Subsection~\ref{ssec:largeprimes}.

\begin{lemma}  \mylabel{lem:small}
For any prime $p$ we have
\begin{equation}  \mylabel{eq:bnd1}
\Pr[\dep_1^p] \leq \left ( \frac{1}{\lambda} \left \lceil \frac{\lambda}{p} \right \rceil\right )^n,\end{equation}
and for any $2\leq m\leq n-1$,
\begin{equation} \label{eq:small}
\Pr[\dep_{m}^p  \mid \neg \dep_{m-1}^p] \leq  
\left ( \frac{1}{\lambda} \left \lceil \frac{\lambda}{p} \right \rceil\right )^{n-m+1}.
\end{equation}
\end{lemma}

\begin{proof}
We have $\dep_1^p$ precisely when the last column of $B$ is zero
modulo $p$. By Lemma~\ref{lem:pdivs}, for any prime $p$ that divides
$2s_n$ we have $\Pr[\dep_1^p]=0$. For a prime $p$ that does not
divide $2s_n$, $\dep_1^p$ is equivalent to the vector
\begin{equation} \mylabel{eq:dfdf} 
\underbracket[.17ex]{(2s_n)^{-1}M_{1..n,n}}_{\mbox{fixed}}+R_{1..n,n} \bmod p \in \Z/(p)^{n \times 1}
\end{equation} being zero modulo $p$. Each random
entry $R_{i,n}$ is equal to $-(2s_n)^{-1}M_{i,n}$ modulo $p$ with
probability at most $$
 \frac{1}{\lambda} \left \lceil \frac{\lambda}{p} \right \rceil.
$$  The bound~(\ref{eq:bnd1})
now follows by noting that vector in~(\ref{eq:dfdf}) has $n$ entries.

Now consider the case $2 \leq m \leq n-1$. By Lemma~\ref{lem:pdivs},
we have that $\Pr[\dep_m^p]=0$ for any prime $p$ that divides
$2s_{n-m+1}$.  Assume henceforth that $p$ does
not divide $2s_{n-m+1}$.  Given $\neg\dep_{m-1}^p$, there 
is an $(m-1)\times (m-1)$ submatrix $D$ in the
last $m-1$ columns of $B$ that is nonsingular modulo $p$.  Assume,
without loss of generality, up to a row permutation of $B$, that $D$
is the trailing $(m-1) \times (m-1)$ submatrix of $B$.  Decompose
the last $m$ columns of $B$ as follows:
$$ \left [ \begin{array}{c|c} v & C \\\hline w & D \end{array} \right ]
\in \Z^{n \times m}.
$$
Then $C$ and $D$ are fixed at this point and vectors $v$ and $w$ still depend
on the random choice of column $n-m+1$ of $R$. Fix the choice of $w$ also.
Note that
$$
\left [ \begin{array}{c|c} I_{n-m+1} & -CD^{-1} \\\hline 
 & D^{-1} \end{array} \right ]
\left [ \begin{array}{c|c} v & C \\\hline w & D \end{array} \right ]
 = 
\left [ \begin{array}{c|c} a \\\hline
\ast & I_{m-1} \end{array} \right ] \bmod p 
\in \Z/(p)^{n \times m}.
$$
Then $\dep_m^p$ is equivalent to the vector
$$
(2s_{n-m+1})^{-1}a = 
\underbracket[.17ex]{(2s_{n-m+1})^{-1}M_{1..n-m+1,n-m+1} -CD^{-1}w}_{\mbox{fixed}}
 + R_{1..n-m+1,n-m+1} 
\bmod p \in \Z/(p)^{(n-m+1) \times 1}
$$
being zero modulo $p$. By a similar argument as before, the probability
of this happening is bounded by~(\ref{eq:small}).
%
\end{proof}

The next lemma follows simply from the union bound on the set of
events for $1\leq i\leq n-1$ that happen when the $i$th column from
the end is the first that is linearly dependent.
\begin{lemma}
\mylabel{lem:cprob}
For any prime $p$ we have  
\begin{equation*} 
\Pr[\dep_{n-1}^p] \leq \Pr[\dep_1^p] + \sum_{i=2}^{n-1} \Pr[\dep_{i}^p  \mid \neg \dep_{i-1}^p].
\end{equation*}
\end{lemma}

\subsection{ Small primes } \mylabel{ssec:smallprimes}

We first deal with the specific small primes $\{3,5,7\}$. Notice that from Proposition~\ref{prp:cARS}, we know that $\Pr[\dep_{n-1}^2 ] = 0$.
\begin{lemma} \mylabel{lem:primes357}  ${\rm Pr}[\vee_{p\in \{3,5,7\}}
\dep_{n-1}^p ] < 0.23$.
\end{lemma}

\begin{proof} 
We exploit the fact that $\lambda$ is
a multiple of $105=3 \times 5 \times 7$.  
Let $p \in \{3,5,7\}$.  Since $p \mid \lambda$, 
the bound of Lemma~\ref{lem:small} simplifies to
\begin{equation*} \mylabel{eq:rrr}
{\rm Pr}[\dep_{m}^p  \mid \neg \dep_{m-1}^p] \leq \
\left ( \frac{1}{p} \right )^{n-m+1},
\end{equation*}
and Lemma~\ref{lem:cprob} gives
\begin{equation}  \mylabel{eq:dkdk}
{\rm Pr}[\dep_{n-1}^p]  \leq  
\sum_{i=1}^{n-1} \left ( \frac{1}{p} \right )^{i+1}
  <  \frac{1}{p} \sum_{i=1}^{\infty} \left ( \frac{1}{p} \right )^{i}
  =    \frac{1}{p(p-1)}.
\end{equation}
Since the events $\dep_{n-1}^3$, $\dep_{n-1}^5$ and $\dep_{n-1}^7$ are
independent,
\begin{equation} \mylabel{eq:kdk}
{\rm Pr}[\vee_{p \in \{3,5,7\}} \dep_{n-1}^p ] = 
1 -  \prod_{p \in \{3,5,7\}}
(1 - {\rm Pr}[\dep_{n-1}^p]).
\end{equation}
The result now follows by
bounding from above  the probabilities
on the right hand size of~(\ref{eq:kdk})
using~(\ref{eq:dkdk}).
\end{proof}

Next we handle the small primes in the range $7 < p < \lambda$.
\begin{lemma} \mylabel{lem:primes7lambda}
${\rm Pr}[\vee_{7 < p < \lambda}\dep_{n-1}^p] < 0.23$
\end{lemma}
\begin{proof} 
Let $7 < p < \lambda$.  Since $p < \lambda$, 
\begin{equation*} 
\frac{1}{\lambda} \left \lceil
\frac{\lambda}{p} \right \rceil
 < 
\frac{1}{\lambda} \left (
\frac{\lambda}{p} + 1\right ) 
 =  \frac{1}{p} + \frac{1}{\lambda} 
 < \frac{2}{p}
  =  \frac{1}{p/2},
\end{equation*}
and the bound of Lemma~\ref{lem:small} simplifies to
\begin{equation} \mylabel{eq:rr2}
{\rm Pr}[\dep_{m}^p \mid \neg \dep_{m-1}^p ]  \leq 
\left (\frac{1}{p/2} \right)^{n-m+1}.
\end{equation}
Lemma~\ref{lem:cprob} together with~(\ref{eq:rr2}) gives
\begin{equation} \mylabel{didi}
{\rm Pr}[\dep_{n-1}^p] \leq \frac{1}{(p/2)(p/2 -1)}
 <  \frac{1}{((p-1)/2)^2}.
\end{equation}
Using the union bound and then~(\ref{didi}) gives
\begin{eqnarray*}
{\rm Pr}[\vee_{7 < p < \lambda} \dep_{n-1}^p] &\leq &
\sum_{7 < p < \lambda} {\rm Pr}[ \dep_{n-1}^p]  \\
& < & \sum_{7 < p < \lambda} \frac{1}{((p-1)/2)^2}\\
& < & \sum_{x \geq 11, {\rm ~odd}} \frac{1}{((x-1)/2)^2} \\
& = & \sum_{x \geq 5} \frac{1}{x^2} \\
& = & \zeta(2) - \sum_{x=1}^{4} \frac{1}{x^2}\\
& = & \frac{\pi^2}{6} - \frac{205}{144}\\
& < & 0.23.
\end{eqnarray*}
\end{proof}

\subsection{Large primes } \mylabel{ssec:largeprimes}

Consider now  the large primes $p \geq \lambda$. Although
it follows from Lemmas~\ref{lem:small} and \ref{lem:cprob} that 
${\rm Pr}[\dep_{n-1}^p] \leq (1/(\lambda(\lambda-1))$ for
any particular prime $p \geq \lambda$, this doesn't help us
to bound ${\rm Pr}[\vee_{p\geq \lambda} \dep_{n-1}^p]$
using the union bound since there exist an infinite number of such primes. Instead, we follow the approach of \citet[Section 6]{EberlyGiesbrechtVillard} and show that we only need to consider those primes which divide some necessarily nonzero minors of $B$.

\begin{lemma} \mylabel{lem:minor}
Any minor of $B$ is bounded in magnitude by $\lambda^{2.5n}$.
\end{lemma}

\begin{proof}
It will suffice to bound $|\det B|$ using Hadamard's inequality, which
states that $|\det B|$ is bounded by the product of the Euclidean
norms of the columns of $B$.
Recall that $B = M + 2RS$ where $M = M \colmod 2S$ and entries in $R$
are chosen from $[0,\lambda-1]$, with $\lambda \geq \max((\det
2S)^{1/n},n)$. Then
\begin{eqnarray*}
|\det B| &\leq & \prod_{j=1}^n \left\| B_{1\ldots n,j}\right\|_2 \\
& = & \prod_{j=1}^n \left\|M_{1\ldots n,j} + 2s_jR_{1\ldots n,j}\right\|_2 \\
& \leq & \prod_{j=1}^n n^{1/2}(2s_j - 1 + 2s_j(\lambda-1)) \\
& < & (\det 2S)n^{n/2}\lambda^n\\
&\leq & \lambda^{2.5n}.
\end{eqnarray*}
\end{proof}

Next we develop the following analogue of Lemma~\ref{lem:small}.
\begin{lemma} \mylabel{lem:largep} We have
$$
{\rm Pr}[\vee_{p \geq \lambda} \dep_1^p] \leq 2.53n \left (\frac{1}{\lambda}\right)^{n-1}
$$
and for any $2 \leq m\leq n-1$, 
$$
{\rm Pr}[\vee_{p \geq \lambda} \dep_m^p \mid \neg 
\vee_{p \geq \lambda} \dep_{m-1}^p] 
\leq 2.53n \left ( \frac{1}{\lambda}\right )^{n-m}.
$$
\end{lemma}
\begin{proof}
By Proposition~\ref{prp:cARS}, $B=M+R(2S)$ is nonsingular modulo 2,
independent of the choice of $R$.  Thus, up to an initial row
permutation of $M$, we may assume that the trailing $j \times j$
submatrix of $B \bmod 2$ is nonsingular over $\Z/(2)$ for every $1\leq j\leq n$.

First consider the case for $m=1$.  Decompose the last column
of $B$ as
$$
\left [ \begin{array}{c} v \\ w \end{array} \right ] \in \Z^{n \times 1},
$$
where $v \in \Z^{(n-1) \times 1}$ and $w \in \Z$.  Fix the choice
of $w$, that is, fix the last entry in the last column
of $R$.  By assumption, $w \neq 0 \bmod 2$ and thus $w\neq 0$ over $\Z$.
For every prime $p \nmid w$ we have ${\rm Pr}[\dep_1^p] = 0$,
and since there are $n-1$ entries in $v$ that are still free
to be chosen, the union bound gives
\begin{eqnarray*}
{\rm Pr}[\vee_{p \geq \lambda} \dep_1^p]  & = & 
{\rm Pr}[\vee_{p \geq \lambda, p \mid w} \dep_1^p] \\
 & \leq & (\log_{\lambda} |w|) \left ( \frac{1}{\lambda} \right )^{n-1}.
\end{eqnarray*}
Lemma~\ref{lem:minor} gives $\log_{\lambda} |w| \leq 2.5n < 2.53n$,
establishing the first part of the lemma.

Now consider $2\leq m \leq n-1$.
Decompose
the last $m$ columns of $B$ as follows:
$$
\left [ \begin{array}{c|c} 
 v & C \\\hline w & D \end{array} \right] \in \Z^{n \times m},
$$
where $D \in \Z^{(m-1) \times (m-1)}$. Then $C$ and $D$ are fixed
at this point and vectors $v$ and $w$ still depend on the random
choice of column $n-m+1$ of $R$.
Let $d=\det D$, which we know to be nonzero.
There are at most $\log_{\lambda} |d|$ primes $p\geq \lambda$ that divide $d$.
Using Lemma~\ref{lem:small} with the union bound gives
\begin{equation} \mylabel{eq:a91}
\sum_{p \geq \lambda, p \mid d} {\rm Pr}[  \dep_m^p \mid \neg \dep_{m-1}^p]
\leq  
(\log_{\lambda} |d|) \left ( \frac{1}{\lambda} \right )^{n-m+1}.
\end{equation}

Next we consider the primes $p \nmid d$.  Note that
$$
\left [ \begin{array}{c|c} dI_{n-m+1} & -dCD^{-1} \\\hline
 & dD^{-1} \end{array} \right ]
\left [ \begin{array}{c|c} 
 v & C \\\hline w & D \end{array} \right] =
\left [ \begin{array}{c|c} a_1 & \\
\vdots & \\
a_{n-m} & \\ 
a_{n-m+1} & \\\hline
\ast & dI_{m-1} \end{array}\right ] \in \Z^{n \times m},
$$
where, by Cramer's rule, $a_{n-m+1}$ is the
determinant of the trailing $m \times m$
submatrix of $B$.  Since $p \nmid d$, event $\dep_m^p$ holds if and
only if the vector
\begin{equation} \mylabel{eq:schur}
 \left [ \begin{array}{c} a_1 \\
\vdots \\ a_{n-m} \\
 a_{n-m+1} \end{array} \right ] =
d \left [ \begin{array}{c} v_1 \\
\vdots \\ v_{n-m} \\
 v_{n-m+1} \end{array} \right ]  - dCD^{-1}w.
\end{equation}
is zero modulo $p$.  Fix the choice of $w$ and $v_{n-m+1}$. Then
$a_{n-m+1} \neq 0$ is also fixed, and for every prime $p\nmid
a_{n-m+1}$ we have $\Pr[\dep_m^p \mid \neg \dep_{m-1}^p]=0$.  Since
there can be at most $\log_{\lambda} |a_{n-m+1}|$ primes $p\geq
\lambda$ that divide $a_{n-m+1}$, and since $v_1,\ldots,v_{n-m}$
are still free to be chosen, we have
\begin{equation} \mylabel{eq:d32}
\sum_{p \geq \lambda, p \nmid d}{\rm Pr}[\dep_{m}^p \mid \neg \dep_{m-1}^p] 
\leq (\log_\lambda |a_{n-m+1}|) \left ( \frac{1}{\lambda}\right)^{n-m}.
\end{equation}
Combining the bounds (\ref{eq:a91}) and (\ref{eq:d32}) and using the
estimate of Lemma~\ref{lem:minor} for $|d|$ and $|a_{n-m+1}|$, we obtain
\begin{eqnarray}
{\rm Pr}[\vee_{p \geq \lambda}\dep_{m}^p \mid \neg 
\vee_{p \geq \lambda} \dep_{m-1}^p ] 
& \leq & 2.5n\left ( \left (\frac{1}{\lambda}\right)^{n-m+1} +
 \left ( \frac{1}{\lambda} \right)^{n-m} \right ) \nonumber \\
 & = &  2.5n \left(\frac{1}{\lambda} \right )^{n-m}\left (
\frac{1}{\lambda} + 1 \right ) \nonumber \\
 &  < &2.53n\left ( \frac{1}{\lambda} \right )^{n-m}. \mylabel{eq:simp3} 
\end{eqnarray}
Here, (\ref{eq:simp3}) follows using $\lambda \geq 105$.
\end{proof}

\begin{lemma} \mylabel{lem:primeslarge}
${\rm Pr}[\vee_{p \geq \lambda} \dep_{n-1}^p ] < 0.03$.
\end{lemma}

\begin{proof}
Analogous to Lemma~\ref{lem:cprob}, we have
$$ 
\Pr[\vee_{p \geq \lambda}\dep_{n-1}^p] \leq \Pr[\vee_{p \geq \lambda} \dep_1^p] + 
\sum_{i=2}^{n-1} \Pr[\vee_{p \geq \lambda} \dep_{i}^p  \mid \neg \vee_{p \geq \lambda}\dep_{i-1}^p].
$$
Using the estimates of Lemma~\ref{lem:largep} now gives
\begin{eqnarray*}
\Pr[\vee_{p \geq \lambda} \dep_{n-1}^p ] &\leq & 
2.53n\left ( \frac{1}{\lambda}\right )^{n-1} + 2.53n \sum_{i=2}^{n-1} \left ( \frac{1}{\lambda}\right)^{n-i} \\
 & < & 2.53n\left ( \frac{1}{\lambda-1} \right ).
\end{eqnarray*}
Simplifying the last bound using the assumption $\lambda \geq 105n$ gives the result.
\end{proof}

\begin{proof}[Proof of Theorem~\ref{thm:cond}]
The probability defined by Theorem~\ref{thm:cond} is bounded by the sum of probabilities in Lemmas~\ref{lem:primes357},~\ref{lem:primes7lambda} and~\ref{lem:primeslarge}, that is,
\begin{align*}
\Pr[\dep_{n-1}] &\leq \Pr[\vee_{p \in\{3, 5, 7\}} \dep_{n-1}^p ] + \Pr[\vee_{7< p < \lambda} \dep_{n-1}^p ] + \Pr[\vee_{p \geq \lambda} \dep_{n-1}^p]\\
&< 0.23+0.23+0.03\\
&<0.5.
\end{align*}
\end{proof}

\section{Almost trivial Hermite form certification} \mylabel{sec:herm}

In this section, we show how to verify whether the last $n-1$ columns of the matrix $B\in\Znn$ from Theorem~\ref{thm:cond} are linearly independent for any prime $p\in\Z$. As we have already mentioned, this means that $B$ is left equivalent to a lower triangular row Hermite form with the shape
\begin{equation} \mylabel{niceherm2}
H=\left [ \begin{array}{cccc}
|\det B| & & & \\
\ast & 1 & & \\
 \vdots & & \ddots & \\
\ast & & & 1 
\end{array} \right ] \in \Z^{n \times n}.
\end{equation}
Our main tool will once more be the Smith form and a Smith massager for $B$.

\begin{theorem} \mylabel{thm:hms}
Let $A\in\Znn$ be nonsingular with Smith form $S$ and a Smith massager $M$. If $H\in\Znn$ is a matrix in Hermite form which satisfies that $\det H=\det S$ and
$HM\equiv 0\colmod S$,
then $H$ is the row Hermite form of $A$.
\end{theorem}

\begin{proof}
The statement follows from Theorem~\ref{thm:canms} and the uniqueness of the Hermite form of $A$.
\end{proof}

We plan to use the description of Theorem~\ref{thm:hms} here in
order to check whether the lower triangular row Hermite form $H$
of the matrix $B$ has $n-1$ trailing trivial columns, and, if yes,
then also compute the first non-trivial column. For this section,
matrices $S$ and $M$ refer to the Smith form and Smith massager of
matrix $B$.

First of all, we need to ensure that the Smith form $S:=\diag(s_1, \ldots, s_n)$ of $B$ also has only one non-trivial invariant factor. If otherwise, then $H$ does not have the desired structure. Let $h_1,h_2, \ldots, h_n$ be the diagonal entries of $H$. The product $h_2\cdots h_n$ equals the $\gcd$ of all the $(n-1)\times (n-1)$ minors in the last $n-1$ columns of $B$. On the other hand, the product $s_1\cdots s_{n-1}$ equals the $\gcd$ of all the $(n-1)\times (n-1)$ minors of $B$, which means that $(s_1\cdots s_{n-1})\mid (h_2\cdots h_n)$. So, if $s_1\cdots s_{n-1}\neq 1$, then $h_2\cdots h_n\neq 1$.

Now, assuming that $S:=\diag(1, \ldots, 1, s_n)$, we are looking to see whether there exists a vector $\bar{h}\in\Z^{(n-1)\times 1}$ such that
\[\left[\begin{array}{cc} s_n & \\ \bar{h} & I_{n-1} \end{array}\right]
M_{1..n,n}\equiv 0 \bmod s_n,\]
which is equivalent to
\begin{equation} \mylabel{eq:mbarhs}
M_{1,n}\bar{h} + M_{2..n,n} \equiv 0 \bmod s_n.
\end{equation}
Since the Hermite form $H$ must be unique, equation (\ref{eq:mbarhs}) must have exactly one solution, which is true if and only if $\gcd(M_{1,n}, s_n)=1$.

The algorithm follows.

\begin{figure}[H]
\begin{center}
\fbox{
\begin{minipage}{\algwidth}
{\tt TrivialLowerHermiteForm}$(B)$\\
\rm
\Input A nonsingular matrix $B\in\Znn$.\\
\Output The lower triangular Hermite form $H\in\Znn$ of $B$ if only the first column is non-trivial, otherwise {\sc NotTrivial}.\\
\Note {\sc Fail} might be returned with probability less than $1/8$.
\begin{enumerate}

\item{} [Compute a Smith massager for $B$.]\\
(If {\tt SmithMassager} fails, return {\sc Fail})\\
$S, M := {\tt SmithMassager}(B)$

\item{} [Certify that $B$ is left equivalent to a matrix $H$ as in (\ref{niceherm2}).]\\
\If $S_{n-1,n-1}\neq 1$ \Then \Return {\sc NotTrivial}\\
\If $\gcd(S_{n,n}, M_{1, n}) \neq 1$ \Then \Return {\sc NotTrivial}

\item{} [Compute matrix $H$ and return.]\\
$H := \left[\begin{array}{cc} h_1 & \\ \bar{h} & I_{n-1} \end{array}\right]$\\
where $h_1 := S_{n,n}$ and $\bar{h} := \Rem(-M_{1,n}^{-1} M_{2..n,n}, S_{n,n})$.\\
\Return $H$
\end{enumerate}
\end{minipage}}
\end{center}
\caption{Algorithm {\tt TrivialLowerHermiteForm} \label{alg:TLHF} }
\end{figure}

\begin{theorem} 
Algorithm {\tt TrivialLowerHermiteForm} is correct and runs in time
\[O(n^{\omega} \, \B(d + \log n)\, (\log n)^2),\]
where $d$ is the average bitlength of the columns of $B\in\Znn$.
\end{theorem}

\begin{proof}
The correctness follows from the preceding discussion.

Regarding the time complexity, the computation of the Smith form $S\in\Znn$ of $B$ along with a Smith massager $M\in\Znn$ dominates the rest of the operations. Let $D_B$ be the partially linearized version of matrix $B$ as
specified by Theorem~\ref{thm:lin}. Then, by Theorem~\ref{thm:smlin},
we can obtain $S$ and $M$ from the Smith form and a Smith massager
for $D_B$ without any extra computation. Therefore, the complexity
of step~1 is bounded by the complexity of computing a Smith massager
for $D_B$, which is $O(n^{\omega} \, \B(d + \log n)\, (\log n)^2)$
by Theorem~\ref{thm:smith}.

The probability of the algorithm failing follows from Corollary~\ref{cor:smith3}.
\end{proof}

\section{A Las Vegas algorithm for Smith form and multipliers} \mylabel{sec:algo}

In this section, we  combine all of the previous results 
established so far in order to develop our multiplier algorithm. In particular, we show that there exists a Las Vegas probabilistic
algorithm that computes the Smith form $S\in\Znn$ of a nonsingular
$A \in \Znn$ along with two unimodular matrices $V, U\in\Znn$ such
that \[AV = US,\] using $O(n^{\omega} \, \B(\log n + \log \|A\|)\,
(\log n)^2)$ bit operations. The algorithm will return the correct
output with probability at least $1/4$ or {\sc Fail} otherwise.

\begin{figure}[H]
\begin{center}
\fbox{
\begin{minipage}{1.01\algwidth}
{\tt SmithFormMultipliers}$(A)$\\
\rm
\Input A nonsingular matrix $A \in \Znn$.\\
\Output The Smith form $S\in\Znn$ of $A$ and two unimodular matrices $U, V\in\Znn$ such that $AV = US$.\\
\Note {\sc Fail} will be returned with probability less than $3/4$.
\begin{enumerate}

\item{} [Compute the Smith form and a Smith massager for $2A$.]\\
(If {\tt SmithMassager} fails, return {\sc Fail})\\
$(2S, M) := {\tt SmithMassager}(2A)$

\item{} [Perturb the Smith massager $M$ by a random matrix.]\\
Pick a uniformly \textbf{random} matrix $R\in\Z/(\lambda)^{n\times n}$ for\\
$\lambda := 105\max(n,\left\lceil(\det 2S)^{1/n}\right\rceil)$ as in Theorem~\ref{thm:cond}.\\
$B:= M + R(2S)$

\item{} [Certify that $B$ is left equivalent to a matrix $H$ as in (\ref{niceherm2}) and return it.]\\
(If {\tt TrivialLowerHermiteForm} fails, return {\sc Fail})\\
$H := {\tt TrivialLowerHermiteForm}(B)$\\
\If $H$ is {\sc NotTrivial} \Then \Return {\sc Fail}

\item{} [Compute a unimodular Smith massager.]\\
$V := BH^{-1}$

\item{} [Compute matrix $U$ and return.]\\
$U := AVS^{-1}$\\
\Return $(S, V, U)$
\end{enumerate}
\end{minipage}}
\end{center}
\caption{Algorithm {\tt SmithFormMultipliers} \label{alg:SFM} }
\end{figure}

\begin{theorem} \mylabel{thm:smithmult}
Algorithm {\tt SmithFormMultipliers} is correct and runs in time
\[O(n^{\omega} \, \B(\log n + \log ||A||)\, (\log n)^2).\]
\end{theorem}

\begin{proof}
Step~1 of the algorithm computes the Smith form and a Smith massager for matrix $2A$. From the Smith form of matrix $2A$ we can trivially obtain the Smith form $S$ of $A$. Furthermore, a Smith massager $M$ for $2A$ is also a Smith massager for $A$ by Lemma~\ref{lem:cA}. Step $1$ runs in $O(n^{\omega} \, \B(\log n + \log \|A\|)\, (\log n)^2)$ by Theorem~\ref{thm:smith}, and it will return {\sc Fail} with probability at most $1/8$ as stated in Corollary~\ref{cor:smith3}.

In step~2, we are perturbing the Smith massager $M$ by a random matrix $R\in\Znn$ multiplied with the Smith form $2S$. By Proposition~\ref{prp:cARS}, matrix $B=M+R(2S)$ is also a Smith massager for $A$, and it is nonsingular. Moreover, by Theorem~\ref{thm:cond}, the last $n-1$ columns of $B$ are linearly independent over $\Z/(p)$ for every prime $p$ with probability greater than $1/2$. As we already mentioned in Section~\ref{sec:random}, this is equivalent to $B$ being left equivalent to a matrix
\begin{equation} \mylabel{niceherm3}
H = \left[\begin{array}{cc} h_1 & \\ \bar{h} & I_{n-1} \end{array}\right],
\end{equation}
where $h_1=|\det B|$.
The runtime of step~2 is dominated by the claimed complexity.

Algorithm {\tt TrivialLowerHermiteForm} called in step~3 then certifies that $B$ has the desired structure and returns matrix $H$. The complexity of the subroutine depends on the average length of the columns of $B$, for which
\[\frac{1}{n}\sum_{j=1}^n\length(B_{1..n,j})\leq \frac{1}{n}\left(\log\left(\prod_{j=1}^n\|B_{1..n,j}\|\right)+n\right) \leq 2.5\log\lambda + 1,\]
as per Lemma~\ref{lem:minor}.
Since $\lambda\in O(n\|A\|)$, the complexity of step~3 is also $O(n^{\omega} \, \B(\log n + \log \|A\|)\, (\log n)^2)$.

Algorithm {\tt TrivialLowerHermiteForm} itself might return {\sc Fail} with probability at most $1/8$. In addition, if it does not fail, the output of the subroutine will be {\sc NotTrivial} with probability at most $1/2$. This makes the probability of success of Algorithm {\tt SmithFormMultipliers} to be at least $1 - (1/8 + 1/2 + 1/8) = 1/4$ as claimed.

Now, since we know that $B\equiv_L H$, the matrix $V:=BH^{-1}$ in step~4 must be integral and unimodular. The evaluation of the product
\[BH^{-1} = B \left[\begin{array}{cc} 1 & \\ -\bar{h} & I_{n-1} \end{array}\right]
\left[\begin{array}{cc} 1/h_1 & \\ & I_{n-1} \end{array}\right]\]
is covered exactly under Lemma~\ref{lem:mwpl} and can be computed,
for $d=n(2.5\log\lambda+1)$, in time $O(n^\omega\, \M(\log n + \log
\|A\|))$.
Furthermore, by Lemma~\ref{lem:masher2}, $V$ is a unimodular Smith massager for $A$.

Finally, by the properties of the Smith massager, matrix $U:=AVS^{-1}$ is integral, and unimodular since $V$ is unimodular. By Lemma~\ref{lem:AMmul}, matrix $U$ can be computed in $O(n^\omega\M(\log n + \log \|A\|))$ bit operations.
\end{proof}

\subsection{Sizes of $V$ and $U$}

It will be important to have good bounds on the magnitude of entries in matrices $V$ and $U$, in order to facilitate the complexity analysis of operations that may use $V$ and $U$ in general.

\begin{lemma} \mylabel{lem:sizes}
The Smith multiplier matrices $V, U\in\Znn$ returned by Algorithm {\tt SmithFormMultipliers} satisfy that:
\begin{itemize}
\item[(i)]  $\|V_{1..n,j}\| \leq cn\|A\| \cdot \left\{\begin{array}{cc} |\det A| + n & \text{ if }j=1 \\ s_j & \text{ otherwise} \end{array}\right.$,
\item[(ii)] $\|U_{1..n,j}\| \leq cn^2\|A\|^2 \cdot \left\{\begin{array}{cc} |\det A| + n & \text{ if }j=1 \\ 1 & \text{ otherwise} \end{array}\right.$.
\end{itemize}
for $c=420$.
\end{lemma}

\begin{proof}
First of all, for $\lambda := 105\max(n,\left\lceil(\det 2S)^{1/n}\right\rceil)$, we have, by Hadamard's bound, that $\lambda\leq 210n\|A\|$.

By construction, we know that $\|B_{1..n, j}\|\leq 2\lambda s_j$ for every $j=1,\ldots,n$. Then, multiplying $B$ with $H^{-1}$ alters only the first column of $B$. The magnitude of the first column of $V=BH^{-1}$ satisfies that
\[\|V_{1..n,1}\|\leq \left( 2\lambda h_1\sum_{j=1}^n s_j\right)/h_1 \leq 2\lambda(|\det A| + n).\]

Furthermore, since $U=AVS^{-1}$, the magnitude of every column of $U$ is bounded by
\[\|U_{1..n,j}\|\leq n\|A\|\|V_{1..n,j}\|/s_j.\]
By replacing $\lambda$ with $210n\|A\|$, the claimed bounds follow.
\end{proof}

\begin{corollary} \mylabel{cor:sizes}
The average bitlength of the columns of both $V$ and $U$ is bounded by $O(\log n + \log \|A\|)$.
\end{corollary}

\subsection{Unbalanced multiplication reduced to balanced}

The remaining tools needed for our algorithm involves reducing unbalanced matrix multiplications to balanced multiplications.
The two lemmas given in this section are used in the proof of Theorem~\ref{thm:smithmult}.
The following lemma is based on \citet[Theorem~20]{BirmpilisLabahnStorjohann19}.

\begin{lemma} \label{lem:mwpl}
Let $M\in\Znn$ and $w\in\Z^{n\times 1}$. If $\sum_{j=1}^n \length(M_{1..n,j}) \leq d$ and $\length(w)\leq d$
for some $d\in\Z_{\geq 0}$, then the product $Mw$ can be computed in time $O(n^\omega\, \M(d/n+\log n))$.
\end{lemma}

\begin{proof}
Choose $X:= 2^{\lceil d/n\rceil}$ and let
\[M = M_0 + M_1X + \cdots + M_{n-1}X^{n-1}\]
\[w = w_0 + w_1X + \cdots + w_{n-1}X^{n-1}\]
be the $X$-adic expansions of $M$ and $w$, respectively. (The coefficients are computed in the symmetric range modulo $X$.) Our approach is to compute the product
\[
\overbrace{\left [ \begin{array}{cccc} M_0 & M_1& \cdots & M_{n-1} 
\end{array} \right ]}^{\textstyle{\bar{M}}}
\overbrace{\left [ \begin{array}{ccccccc} w_0 & w_1 & \cdots & w_{n-1} &&& \\
 & w_0 & \cdots & w_{n-2} & w_{n-1} && \\
 & & \ddots & \vdots & \vdots & \ddots & \\
 & & & w_0 & w_1 & \cdots  & w_{n-1} \end{array} \right ]}^{\textstyle{\bar{W}}},
\]
from which $Mw$ can be recovered fast. (Notice that the operations to compute the $X$-adic expansion from a matrix or the matrix from an $X$-adic expansion take linear time on the number of entries when $X$ is a power of $2$.)

Now, the column dimension of $\bar{M}$ and row dimension of $\bar{W}$ is $n^2$ which is too large to fit within our target complexity. However, because of the assumption that $\sum_{j=1}^n \length(M_{1..n,j}) \leq d$ and the fact that $\log(X)=\lceil d/n\rceil$, matrix $\bar{M}$ must contain many zero columns. More specifically, the number of non-zero columns in $\bar{M}$ cannot exceed
\[\sum_{j=1}^n \left\lceil\frac{\length(M_{1..n,j})}{\lceil d/n\rceil}\right\rceil \leq \sum_{j=1}^n \left( n\frac{\length(M_{1..n,j})}{d} + 1 \right)\leq 2n.\]

Therefore, let $\tilde{M}\in\Z ^{n\times 2n}$ be the matrix obtained from $\bar{M}$ by omitting $n^2-2n$ zero columns, and let $\tilde{W}\in\Z^{2n\times 2n-1}$ be the matrix obtained from $\bar{W}$ by omitting $n^2-2n$ rows corresponding to the columns that were omitted in $\bar{M}$. This transformation reduces the multiplication of $\bar{M}\bar{W}$ to the multiplication of $\tilde{M}\tilde{W}$ which can be done in time $O(n^\omega\, \M(d/n+\log n))$ since $\log \|\tilde{M}\tilde{W}\|\in O(d/n +\log n)$.
\end{proof}

Moreover, the following lemma uses a similar proof technique and
is based on \citet[Lemma~19]{BirmpilisLabahnStorjohann20}.

\begin{lemma} \mylabel{lem:AMmul}
Let $A,M\in\Znn$.  If $\length(A)\leq d$ and $\sum_{j=1}^n \length(M_{1..n,j}) \leq nd$
for some $d\in\Z_{\geq 0}$, then we can compute the product $AM$ in time
 $O(n^\omega\, \M(d+\log n))$.
\end{lemma}

\begin{proof}
Choose $X:= 2^d$ and let
\[M = M_0 + M_1X + \cdots + M_{n-1}X^{n-1}\]
be the $X$-adic expansion of $M$. (The coefficients are computed in the symmetric range modulo $X$.) Our approach is to compute the product
\[
A\overbrace{\left [ \begin{array}{cccc} M_0 & M_1& \cdots & M_{n-1} 
\end{array} \right ]}^{\textstyle{\bar{M}}},
\]
from which $AM$ can be recovered fast. (Notice that the operations to compute the $X$-adic expansion from a matrix or the matrix from an $X$-adic expansion take linear time on the number of entries when $X$ is a power of $2$.)

Now, the column dimension of $\bar{M}$ is $n^2$ which is too large to fit within our target complexity. However, because of the assumption that $\sum_{j=1}^n \length(M_{1..n,j}) \leq nd$ and the fact that $\log(X)=d$, matrix $\bar{M}$ must contain many zero columns. More specifically, the number of non-zero columns in $\bar{M}$ cannot exceed
\[\sum_{j=1}^n \left\lceil\frac{\length(M_{1..n,j})}{d}\right\rceil \leq \sum_{j=1}^n \left( \frac{\length(M_{1..n,j})}{d} + 1 \right)\leq 2n.\]

Therefore, let $\tilde{M}\in\Z ^{n\times 2n}$ be the matrix obtained from $\bar{M}$ by omitting $n^2-2n$ zero columns. This transformation reduces the multiplication of $A\bar{M}$ to the multiplication of $A\tilde{M}$ which can be done in time $O(n^\omega\M(d+\log n))$ since $\log \|A\tilde{M}\|\in O(d +\log n)$.
\end{proof}

\begin{remark} \label{lem:AMmulU}
Lemma~\ref{lem:AMmul} can be also stated with matrix $A\in\Znn$ replaced by a matrix $U\in\Znn$ that satisfies
$\sum_{j=1}^m \length(U_{1..n,j})\leq nd$.
\end{remark}

\section{Application: Computing an outer product adjoint formula
for $A$} \mylabel{sec:opa}

In this section, we mention an application of the Smith form with
the multiplier matrices.  Let $A \in \Z^{n \times n}$ be nonsingular
and assume that we have precomputed the Smith form $S$ of $A$,
together with unimodular matrices $U$ and $V$ such that $AV = US$.

Let $s:=S_{n,n}$ be the largest invariant factor of $A$.  Recall
that $s$ is the minimal positive integer that clears the denominators
in $A^{-1} \in \Q^{n \times n}$, that is, if entries in $A^{-1}$ are expressed as reduced
fractions, then $s$ is the least common multiple of the denominators 
of the entries.  The inverse of $A$ can thus be recovered by
computing the integer matrix $sA^{-1}$ and dividing by $s$.
As a tool to compute $A^{-1}$, \citet{Storjohann10a} developed an
algorithm to compute an \emph{outer
product adjoint formula} for $A$: a triple of matrices $(\bar{V},
S, \bar{U})$ such that \[\Rem(sA^{-1}, s) = \Rem(\bar{V}(sS^{-1})\bar{U},
s).\]
Moreover, $\bar{V} = (\bar{V} \colmod S)$ and $\bar{U} = (\bar{U}
\rowmod S)$, where $\bar{U} \rowmod S$ means reduction of the rows
modulo the corresponding diagonal entries of $S$.
While a tight upper bound for the number of bits
required to represent $\Rem(sA^{-1}, s)$ explicitly in the worst case
is $O(n^3 (\log n + \log \|A\|))$, an outer
product adjoint formula $(\bar{V}, S, \bar{U})$ requires only
$O(n^2 (\log n + \log \|A\|))$ bits.
Note that $\Rem(sA^{-1},s)/s$ corresponds to only the fractional part of $A^{-1}$,
that is, if $C$ is the matrix obtained from $\Rem(sA^{-1},s)$ by reducing entries
in the symmetric range modulo $s$, then
$A^{-1} - C/s \in \Z^{n \times n}$ may be nonzero.
However, if $\|A^{-1}\| < 1/2$, 
then $C$ will be identically equal to $sA^{-1}$.

\begin{example} \mylabel{ex:opa}
Matrix
$$
A =  \left[ \begin {array}{cccc} -6&3&-13&-15\\ 
-4&19&12 &-1\\ -4&10&-6&17\\ -26&-13&1&-2 \end {array} \right] 
$$
has Smith form $S := {\rm Diag}(s_1,s_2,s_3,s_4) = {\rm Diag}(1,1,9,29088)$ 
and 
$$
s_4 A^{-1} =  \left[ \begin {array}{cccc} -271&-402&-373&-937\\ 
580&920&524&-356\\ -1074&804&-870&258
\\ -784&-352&1008&80\end {array} \right] . $$
An outer product adjoint formula for $A$ is given by $(\bar{V},S,\bar{U})$
where 
$$
\bar{V} =  \left[ \begin {array}{cccc} 0&0&7&805\\ 0&0&5&23668
\\ 0&0&3&6\\ 0&0&4&10224 \end {array} \right] 
\mbox{~~and~~}
\bar{U} = 
 \left[ \begin {array}{cccc} 0&0&0&0\\ 0&0&0&0
\\ 2&2&0&2\\ 20829&1750&28943& 16203\end {array} \right] .
$$
For this particular $A$, which satisfies $||A^{-1}|| < 1/2$, multiplying out
$\bar{V}(s_4 S^{-1}) \bar{U}$ and reducing entries in the symmetric
range modulo $s_4$ gives $s_4 A^{-1}$. Because $s_1=s_2=1$ the first
two columns of $V$ and first two rows of $U$ can be omitted, giving
$$
 \left[ \begin {array}{cc} 7&805\\ 5&23668
\\ 3&6\\ 4&10224\end {array} \right] 
 \left[ \begin {array}{cc} 3232&\\ &1\end {array} \right] 
 \left[ \begin {array}{cccc} 2&2&0&2\\ 20829&1750& 28943&16203\end {array}
\right]  \equiv s_4A^{-1} \bmod s_4.$$
\end{example}

There is a direct relationship between an outer product adjoint
formula and the unimodular Smith multipliers $U$ and $V$.

\begin{lemma} \mylabel{lem:dkdk}
Let $U,V \in \Znn$ be unimodular matrices such that $AV = US$. Then, the triple $(V \colmod S, S, U^{-1} \rowmod S)$ gives an outer product adjoint formula for $A$.
\end{lemma}

\begin{proof}
We have that $sA^{-1} = V(sS^{-1})U^{-1}$. Furthermore,  $V(sS^{-1})=(V\colmod S)(sS^{-1})\bmod s$ and $(sS^{-1})U^{-1}=(sS^{-1})(U^{-1}\rowmod S)\bmod s$, and so 
\[\Rem(sA^{-1}, s) = \Rem((V\colmod S)(sS^{-1})(U^{-1}\rowmod S), s).\]
\end{proof}

\cite{Storjohann10a} gives a randomized algorithm to compute an outer product
adjoint formula in 
\begin{equation} \mylabel{eq:cost1}
O(n^2(\log n) \B(n(\log n + \log \|A\|))
\end{equation}
plus 
\begin{equation} \mylabel{eq:cost2}
O(n^3 \max(\log n, \log \|A^{-1}\|)\, \B(\log n + \log \|A\|))
\end{equation}
bit operations.
Note that~(\ref{eq:cost1})
implies that fast (pseudo-linear) integer arithmetic  needs to be
used to achieve a cost that is softly cubic in $n$, while
(\ref{eq:cost2}) reveals a sensitivity to $\|A^{-1}\|$.
Indeed, we may have $\log\|A^{-1}\| \in \Omega(n(\log n+ \log \|A\|))$, in which case
the upper bound in~(\ref{eq:cost2}) becomes quartic in $n$.  It was
left as an open question if an outer product adjoint formula can
be computed in time $(n^{\omega} \log \|A\|)^{1+o(1)}$ bit operations.
Here, we can resolve this question by using the approach of
Lemma~\ref{lem:dkdk}.

\begin{theorem}  \mylabel{thm:opa}
Assume we have the output $(S, V, U)$ of Algorithm ${\tt
SmithFormMultipliers}(A)$. Then, an outer product adjoint formula
for $A$ can be computed in time $O(n^{\omega} \, \M(\log n + \log
\|A\|)\log n)$.
\end{theorem}

\begin{proof} 
First compute $\bar{V} := V \colmod S$. This can be done in time $O(n
\sum_{i=1}^n \M(\length(V_{1\ldots n,i}))$. By Corollary~\ref{cor:sizes},
$\sum_{i=1}^n \length(V_{1\ldots n,i}) \in O(n(\log n + \log \|A\|))$,
which shows that the matrix $\bar{V}$ can be computed in time $O(n\, 
\M(n(\log n + \log \|A\|))$.

It remains to compute $\bar{U} := U^{-1} \rowmod S$. Let $D \in
\Z^{m \times m}$ be the partial column linearization of $U$ as in
Theorem~\ref{thm:lin}. It will be that $m \in O(n)$, and again by
Corollary~\ref{cor:sizes}, $\log \|D\| \in O(\log n + \log \|A\|)$.
Therefore, by Lemma~\ref{lem:dpol}, we can compute a straight line
formula for $D^{-1}$ in time $O(n^{\omega} \, \M(\log n + \log
\|A\|)\log n)$. The formula consists of $O(\log n)$ integer matrices
of dimension $m$ and bitlength bounded by $O(\log n + \log \|A\|)$.

Finally, we can compute $U^{-1}\rowmod S$ by evaluating $D^{-1}\rowmod
\diag(S,I_{m-n})$ using the straight line formula. The evaluation
of the formula requires $O(\log n)$ matrix multiplications where
the first operand is an $m\times m$ integer matrix reduced
$\rowmod\diag(S,I)$ and the second operand is an $m\times m$ integer
matrix with bitlength bounded by $O(\log n + \log \|A\|)$. This
type of matrix multiplication falls exactly under Lemma~\ref{lem:AMmul}
by simply transposing the operation. Therefore, we can compute
$U^{-1}\rowmod S$ in time $O(n^{\omega} \, \M(\log n + \log \|A\|)\log
n)$.
\end{proof}

An application of the outer product adjoint formula is to compute
the proper fractional part of a linear system solution.  Let $b \in
\Z^{n \times 1}$ satisfy $\log ||b|| \in (n \log ||A||)^{1+o(1)}$.
Then 
$$
A^{-1}b = \overbrace{A^{-1}b - \Rem(sA^{-1}b, s)/s}^{\textstyle \in \Z^n} + \Rem(sA^{-1}b, s)/s,
$$
where $\Rem(sA^{-1}b,s)/s$ is a vector of proper fractions.  By
Lemma~\ref{thm:solve}, $A^{-1}b \in \Q^{n \times 1}$ can be computed
in a Las Vegas fashion in $(n^{\omega} \log ||A||)^{1+o(1)}$ bit
operations, or $(n^3 \log ||A||)^{1+o(1)}$ bit operations if $\omega=3$.
If an outer product adjoint formula for $A$ is known, then the
proper fractional part of $A^{-1}b$ can be computed in only $(n^2
\log ||A||)^{1+o(1)}$ bit operations.
The following result is a corollary of~\citep[Lemma~4.11]{Storjohann10a}.
\begin{lemma} \mylabel{lem:mindenom}
Assume we have an outer product adjoint formula $(\bar{V},S,\bar{U})$
for a nonsingular $A \in \Znn$, and let $s = S_{n,n}$. Given a
vector $b \in \Z^{n \times 1}$ with $\log \|b\|\in O(\log s)$, we
can compute $\Rem(sA^{-1}b, s)$ in time $O(n\, \M(\log s))$.
\end{lemma}
\begin{example}
Let $A \in \Z^{n \times n}$ be the matrix of Example~\ref{ex:opa}
and
$$
b =  \left[ \begin {array}{c} 25\\ 94 \\ 12\\ -2\end {array} \right].
$$
Then
$$
\bar{V}(29088 S^{-1})\bar{U}b \equiv \left[ \begin {array}{c} 11011\\ \noalign{\medskip}20716
\\ \noalign{\medskip}8682\\ \noalign{\medskip}17424\end {array}
\right ] \bmod 29088.
$$ Indeed, we have
$$
A^{-1}b =  
 \left[ \begin {array}{c} -2\\ \noalign{\medskip}3
\\ \noalign{\medskip}1\\ \noalign{\medskip}-2\end {array} \right]  +
\left[ \begin {array}{c} 11011\\ \noalign{\medskip}20716
\\ \noalign{\medskip}8682\\ \noalign{\medskip}17424\end {array}
 \right] \frac{1}{29088}.
$$
\end{example}

Applying Lemma~\ref{lem:mindenom} with $b=I_n$ gives the following corollary
of Theorems~\ref{thm:smithmult} and~\ref{thm:opa}.

\begin{corollary} 
\mylabel{cor:inv}
Given a nonsingular integer input matrix $A \in \Z^{n \times n}$, the largest invariant factor $s$ of $A$, together with 
$\Rem(sA^{-1},s)$, can be computed in a Las Vegas fashion
in 
$$O(n^{\omega}\, \B(\log n + \log ||A||)(\log n)^2 + n^2\, \M(\log s))$$
bit operations. This is bounded by $(n^3 \log ||A||)^{1+o(1)}$ bit operations.
\end{corollary}

\section{ Conclusion and topics for future research}

In this paper we have  presented a new, Las Vegas probabilistic
algorithm for determining the unimodular Smith multipliers for a
nonsingular integer matrix.   Combining this with our previous
results in \citep{BirmpilisLabahnStorjohann20}, implies that we
can determine the  Smith form  and a pair of unimodular multipliers
in time $(n^\omega\log \|A\|)^{1+o(1)}$, approximately  about  same
number of bit operations as required to multiply two matrices of
the same dimension and size of entries as the input matrix. We have
also given explicit bounds on the sizes of our multipliers and made use of
such bounds to efficiently determine an outer adjoint formula for
an integer matrix. We also include computational tools and partial
linearization sections which should be of independent interest.

In terms of future directions, a natural direction is to find a
{\em{deterministic}} algorithm for both the Smith form and the Smith
form with multipliers problems.  In the case of integer matrices
we have already seen that linear system solving can be derandomized
within the desired cost.  An easier problem than to derandomize
Smith form computation would be to first find a deterministic
algorithm for finding only the largest invariant factor $s_n$, a
problem that has a solution in the case of polynomial
matrices \citep{ZhouLabahnStorjohann14}.

Another problem which arises naturally is that of finding  algorithms
for the computation of other integer matrix forms, in particular
the Hermite normal form, with the target complexity being the number
of bit operations required to multiply two matrices of the same
dimension and size of entries as the input matrix. We expect that
our primary tool, the Smith massager can also play an important
intermediate role here.

Finally, our algorithms and tools all assume that the input matrix
is nonsingular, unlike for example the procedure from
\cite{KaltofenVillard04a}. It is of interest to extend the present
work to singular integer matrices, likely through compression
techniques to reduce the problem to a smaller nonsingular matrix.

\section{Acknowledgements}

We would like to thank the anonymous referees for their suggestions on making the paper more readable.
This research was partly supported by the Natural Sciences and Engineering Research Council (NSERC) Canada.

\newcommand{\SortNoop}[1]{}

\end{document}